%% file: paper.tex
\documentclass[11pt]{article}

\usepackage{multirow}

\usepackage[utf8]{inputenc} 
\usepackage[T1]{fontenc}    
\usepackage{hyperref}       
\usepackage{url}            
\usepackage{booktabs}       
\usepackage{amsfonts}       
\usepackage{nicefrac}       
\usepackage{microtype}      
\usepackage{graphicx}
\usepackage[font=bf]{caption}
\usepackage{color, url}
\usepackage{xspace}
\usepackage{subfloat}
\usepackage{subfigure}
\usepackage{amsmath}
\usepackage{mathrsfs}
\usepackage{amssymb}
\usepackage{amsmath}
\usepackage{amsthm}
\usepackage{epstopdf}
\usepackage{balance}
\newtheorem{theorem}{Theorem}
\newtheorem{lemma}{Lemma}
\newtheorem{definition}{Definition}
\newtheorem{corollary}{Corollary}
\usepackage{algorithm}
\usepackage{algorithmic}

\usepackage{bbm}
\usepackage{epsfig,endnotes}
\usepackage{grffile}

\newcounter{subcopyrightbox@save}
\usepackage{bm}
\usepackage{mathtools}

\usepackage{cleveref}
\usepackage{geometry}
\usepackage{changepage}

\allowdisplaybreaks
\usepackage[sort&compress,round,comma,authoryear]{natbib}

\usepackage{algorithmic}

\newcommand{\myparatight}[1]{\smallskip\noindent{\bf {#1}:}~}

\newcommand{\RN}[1]{%
  \textup{\uppercase\expandafter{\romannumeral#1}}%
}

\DeclareMathOperator*{\argmax}{arg\,max}

\AtBeginDocument{%
  \providecommand\BibTeX{{%
    \normalfont B\kern-0.5em{\scshape i\kern-0.25em b}\kern-0.8em\TeX}}}

\begin{document}

\begin{center}
{\LARGE{\bf{Intrinsic Certified Robustness of Bagging against Data Poisoning Attacks}}}

  \vspace{1cm}
\begin{tabular}{ccccc}
Jinyuan Jia & & Xiaoyu Cao && Neil Zhenqiang Gong \\
Duke University && Duke University && Duke University \\
jinyuan.jia@duke.edu && xiaoyu.cao@duke.edu &&neil.gong@duke.edu \\ 
\end{tabular}
\end{center}
  \vspace{.5cm}

\input{abstract}
\input{introduction.tex}

\input{method.tex}

\input{practice.tex}
\input{evaluation.tex}

\input{related}
\input{conclusion}

\section{ Acknowledgments}
We thank the anonymous reviewers for insightful reviews. 
This work was supported by NSF grant No. 1937786.

{
\bibliographystyle{plainnat}
\bibliography{refs,refs1,refs2}
}

\input{appendix}

\end{document}

%% file: abstract.tex
\begin{abstract}

In a \emph{data poisoning attack}, an attacker modifies, deletes, and/or inserts some training examples to corrupt the learnt machine learning model. \emph{Bootstrap Aggregating (bagging)} is a well-known ensemble learning method,  which trains multiple base models on random subsamples of a training dataset using a base learning algorithm and uses majority vote to predict labels of testing examples. We prove the intrinsic certified robustness of bagging against data poisoning attacks. Specifically, we show that bagging with an arbitrary base learning algorithm provably predicts the same label for a testing example when the number of modified, deleted, and/or inserted training examples is bounded by a threshold. Moreover, we show that our derived threshold is tight if no assumptions on the base learning algorithm are made.  We evaluate our method on  MNIST and CIFAR10. For instance, our method achieves a certified accuracy of $91.1\%$ on MNIST when arbitrarily modifying, deleting, and/or inserting 100 training examples. Code is available at: \url{https://github.com/jjy1994/BaggingCertifyDataPoisoning}.
\end{abstract}

%% file: introduction.tex
\section{Introduction}
Machine learning models trained on user-provided data are vulnerable to \emph{data poisoning attacks}~\citep{nelson2008exploiting,biggio2012poisoning,xiao2015feature,li2016data,steinhardt2017certified,shafahi2018poison}, in which malicious users carefully poison (i.e., modify, delete, and/or insert) some training examples such that  the learnt model is corrupted and makes predictions for testing examples as an attacker desires. In particular, the corrupted model predicts incorrect labels for a large fraction of testing examples indiscriminately (i.e., a large testing error rate) or for some attacker-chosen testing examples. Unlike adversarial examples~\citep{szegedy2013intriguing,carlini2017towards}, which carefully perturb each testing example such that a model predicts an incorrect label for the perturbed testing example, data poisoning attacks corrupt the model such that it predicts incorrect labels for many clean testing examples. Like adversarial examples, data poisoning attacks pose severe security threats to machine learning systems.

To mitigate data poisoning attacks, various defenses~\citep{Cretu08,barreno2010security,Suciu18,Tran18,Feng14,Jagielski18,ma2019data,wang2020certifying,rosenfeld2020certified} have been proposed in the literature. Most of these defenses~\citep{Cretu08,barreno2010security,Suciu18,Tran18,Feng14,Jagielski18} achieve \emph{empirical} robustness against certain data poisoning attacks and are often broken by strong adaptive attacks. To end the cat-and-mouse game between attackers and defenders, \emph{certified defenses}~\citep{ma2019data,wang2020certifying,rosenfeld2020certified} were proposed. We say a learning algorithm is certifiably robust against data poisoning attacks if it can learn a classifier that provably predicts the same label for a testing example when the number of poisoned training examples is bounded.  
For instance, \cite{ma2019data} showed that a classifier trained with differential privacy certifies robustness against data poisoning attacks. \cite{wang2020certifying} and \cite{rosenfeld2020certified} leveraged \emph{randomized smoothing}~\citep{cao2017mitigating, cohen2019certified}, which was originally designed to certify robustness against adversarial examples, to certify robustness against data poisoning attacks that modify  labels and/or features  of existing training examples. 


However, these certified defenses suffer from two major limitations. First, they are only applicable to limited scenarios, i.e., \cite{ma2019data} is limited to learning algorithms that can be differentially private, while \cite{wang2020certifying} and \cite{rosenfeld2020certified} are limited to data poisoning attacks that only modify existing training examples. Second, their certified robustness guarantees are loose, meaning that a learning algorithm is certifiably more robust than their guarantees indicate. We note that \cite{steinhardt2017certified} derives an approximate upper bound of the loss function for data poisoning attacks. However, their method cannot certify that the learnt model predicts the same label for a testing example.

We aim to address these limitations in this work. Our approach is based on a well-known ensemble learning method called \emph{Bootstrap Aggregating (bagging)}~\citep{breiman1996bagging}. Given a training dataset, we create a random \emph{subsample} with $k$ training examples sampled from the training dataset uniformly at random with replacement. Moreover, we use a deterministic or randomized base learning algorithm to learn a base classifier on the subsample. Due to the randomness in sampling the subsample and the (randomized) base learning algorithm, the label predicted for a testing example $\mathbf{x}$ by the learnt base classifier is random. Therefore, we define $p_j$ as the probability that the learnt base classifier predicts label $j$ for $\mathbf{x}$, where $j=1,2,\cdots,c$. We call $p_j$ \emph{label probability}. In bagging, the \emph{ensemble classifier} essentially predicts the label with the largest label probability for $\mathbf{x}$.


Our first major theoretical result is that we prove the ensemble classifier in bagging predicts the same label for a testing example when the number of poisoned training examples is no larger than a threshold. We call the threshold \emph{certified poisoning size}. Our second major theoretical result is that we prove our derived certified poisoning size is tight (i.e., it is impossible to derive a certified poisoning size larger than ours) if no assumptions on the base learning algorithm are made. Note that the certified poisoning sizes may be different for different testing examples. 

Our certified poisoning size for a testing example  is the optimal solution to an optimization problem, which involves the testing example's largest and second largest label probabilities predicted by the bagging's ensemble classifier. 
However, it is computationally challenging to compute the exact largest and second largest label probabilities, as there are an exponential number of subsamples with $k$ training examples. To address the challenge, we propose a Monto Carlo algorithm to simultaneously estimate a lower bound of the largest label probability and an upper bound of the second largest label probability for multiple testing examples via training $N$ base classifiers on $N$ random subsamples. Moreover, we design an efficient algorithm to solve the optimization problem with the estimated largest and second largest label probabilities to compute certified poisoning size. 

We empirically evaluate our method on MNIST and CIFAR10. For instance, our method can achieve a certified accuracy of $91.1\%$ on MNIST when 100 training examples are arbitrarily poisoned, where $k=100$ and $N=1,000$. Under the same attack setting, \cite{ma2019data}, \cite{wang2020certifying}, and \cite{rosenfeld2020certified} achieve $0$ certified accuracy on a simpler MNIST 1/7 dataset. Moreover, we show that training the base classifiers using transfer learning can significantly improve the certified accuracy.

Our contributions are summarized as follows: 
\begin{itemize}
    \item We derive the first intrinsic certified robustness of bagging against data poisoning attacks and prove the tightness of our robustness guarantee.  
    
    \item We develop algorithms to compute the certified poisoning size in practice. 
    
    \item We evaluate our method on MNIST and CIFAR10. 
\end{itemize}

All our proofs are shown in the Supplemental Material.

%% file: method.tex
\section{Certified Robustness of Bagging}

 Assuming we have a training dataset $\mathcal{D}=\{(\mathbf{x}_1,y_1), (\mathbf{x}_2,y_2), \cdots, (\mathbf{x}_n,y_n)\}$ with $n$ examples, where $\mathbf{x}_i$ and $y_i$ are  the feature vector and  label of the $i$th training example, respectively. Moreover, we are given an arbitrary deterministic or randomized base learning algorithm $\mathcal{A}$, which takes a training dataset $\mathcal{D}$ as input and outputs a classifier $f$, i.e., $f=\mathcal{A}(\mathcal{D})$. $f(x)$ is the predicted label for a testing example $\mathbf{x}$. For convenience, we  jointly represent the training and testing processes as $\mathcal{A}(\mathcal{D}, \mathbf{x})$, which is $\mathbf{x}$'s label predicted by a classifier that is trained using algorithm  $\mathcal{A}$ and  training dataset $\mathcal{D}$.

\myparatight{Data poisoning attacks} In a data poisoning attack, an attacker poisons the training dataset $\mathcal{D}$ such that the learnt classifier makes predictions for testing examples as the attacker desires. In particular, the attacker can carefully \emph{modify}, \emph{delete}, and/or \emph{insert} some training examples in $\mathcal{D}$ such that $\mathcal{A}(\mathcal{D}, \mathbf{x})\neq \mathcal{A}(\mathcal{D}', \mathbf{x})$ for many testing examples $\mathbf{x}$ or some attacker-chosen $\mathbf{x}$, where  $\mathcal{D}'$ is the poisoned training dataset. We note that modifying a training example means modifying its feature vector and/or label.  
 We denote the set of poisoned training datasets with at most $r$ poisoned training examples as follows:
\begin{align}
{B}(\mathcal{D},r)=\{\mathcal{D}'|\max\{|\mathcal{D}|,|\mathcal{D}'|\}-|\mathcal{D}\cap\mathcal{D}^{\prime}|\leq r\}.
\end{align}
 Intuitively, $\max\{|\mathcal{D}|,|\mathcal{D}'|\}-|\mathcal{D}\cap\mathcal{D}^{\prime}|$ is the minimum number of modified/deleted/inserted training examples that can change $\mathcal{D}$ to $\mathcal{D}'$.  

\myparatight{Bootstrap aggregating (Bagging)~\citep{breiman1996bagging}} Bagging is a well-known ensemble learning method. Roughly speaking, bagging  creates many subsamples of a training dataset with replacement and trains a base classifier on each subsample. For a testing example, bagging uses each base classifier to predict its label and takes majority vote among the predicted labels as the  label of the testing example. Figure~\ref{fig:illustration} shows a toy example to illustrate why bagging certifies robustness against data poisoning attacks. When the poisoned training examples are minority in the training dataset, a majority of the subsamples do not include any poisoned training examples. Therefore, a majority of the base classifiers and the bagging's predicted labels for testing examples are not influenced by the poisoned training examples. 

\begin{figure}[!t]
    \center
   {\includegraphics[width=0.7\textwidth]{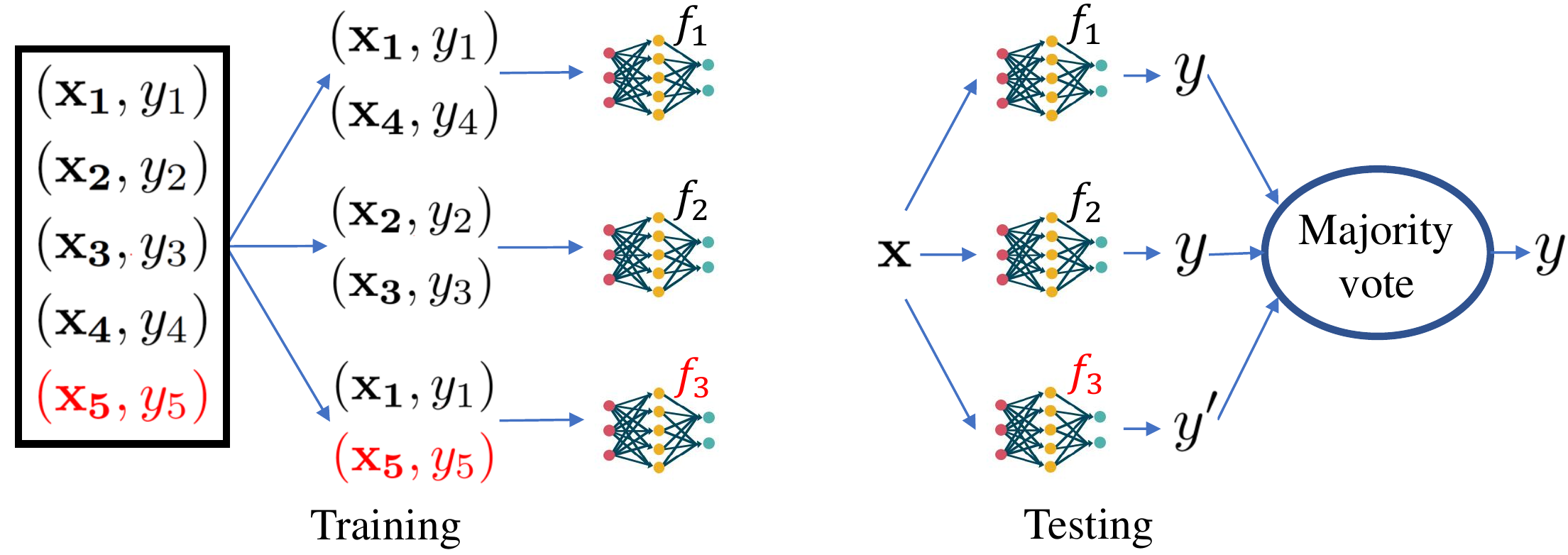}}  
    \caption{An example to illustrate why bagging is robust against data poisoning attacks, where ($\mathbf{x}_5$,$y_5$) is the poisoned training example. Base classifiers $f_1$ and $f_2$ are trained using clean training examples and bagging predicts the correct label for a testing example after majority vote among the three base classifiers. }
    \label{fig:illustration}
\end{figure}

Next, we describe a probabilistic view of bagging, which makes it possible to theoretically analyze its certified robustness against data poisoning attacks. Specifically, we denote by $g(\mathcal{D})$ a random subsample, which is a list of $k$ examples that are sampled from $\mathcal{D}$ with replacement uniformly at random. We use the base learning algorithm $\mathcal{A}$ to learn a base classifier on $g(\mathcal{D})$. 
Due to the randomness in sampling the subsample $g(\mathcal{D})$ and  the (randomized) base learning algorithm $\mathcal{A}$, the label $\mathcal{A}(g(\mathcal{D}), \mathbf{x})$ predicted by the base classifier learnt on $g(\mathcal{D})$ for $\mathbf{x}$ is random. We denote by $p_j=\text{Pr}(\mathcal{A}(g(\mathcal{D}), \mathbf{x})=j)$ the probability that the learnt base classifier predicts label $j$ for $\mathbf{x}$, where $j=1,2,\cdots,c$. We call $p_j$ \emph{label probability}. The \emph{ensemble classifier} $h$ in bagging essentially predicts the label with the largest label probability for $\mathbf{x}$, i.e., we have: 
\begin{align}
\label{baggingdefinition}
h(\mathcal{D}, \mathbf{x})=\argmax_{j\in\{1,2,\cdots,c\}}p_j,
\end{align}
where  $h(\mathcal{D}, \mathbf{x})$ is the predicted label for $\mathbf{x}$ when the ensemble classifier $h$ is trained on $\mathcal{D}$.


\myparatight{Certified robustness of bagging}  We prove the certified robustness of bagging against data poisoning attacks. In particular, we show that the ensemble classifier in bagging predicts the same label for a testing example when the number of poisoned training examples is no larger than some threshold (called \emph{certified poisoning size}). Formally, we aim to show $h(\mathcal{D}^{\prime},\mathbf{x})=h(\mathcal{D},\mathbf{x})$ for $\forall \mathcal{D}^{\prime} \in {B}(\mathcal{D},r^{*})$, where $r^{*}$ is the certified poisoning size. For convenience, we define the following two random variables: 
\begin{align}
 X = g(\mathcal{D}),  Y = g(\mathcal{D}^{\prime}), 
\end{align}
where $X$ and $Y$  are two random subsamples with $k$ examples sampled from $\mathcal{D}$ and $\mathcal{D}^{\prime}$ with replacement uniformly at random, respectively. $p_j=\text{Pr}(\mathcal{A}(X, \mathbf{x})=j)$ and $p_j'=\text{Pr}(\mathcal{A}(Y, \mathbf{x})=j)$ are the label probabilities of label $j$ for testing example $\mathbf{x}$ when the training dataset is $\mathcal{D}$ and its poisoned version $\mathcal{D}^{\prime}$, respectively. For simplicity, we use $\Omega$ to denote the joint space of $X$ and $Y$, i.e., each element in $\Omega$ is a subsample of $k$ examples sampled from $\mathcal{D}$ or $\mathcal{D}^{\prime}$ uniformly at random with replacement.

Suppose the ensemble classifier predicts label $l$ for $\mathbf{x}$ when trained on the clean training dataset, i.e., $h(\mathcal{D},\mathbf{x})=l$. Our goal is to  find the maximal poisoning size $r$ such that the ensemble classifier still predicts label $l$ for $\mathbf{x}$ when trained on the poisoned training dataset with at most $r$ poisoned training examples. Formally, our goal is to find the maximal poisoning size $r$ such that the following inequality is satisfied for $\forall \mathcal{D}^{\prime} \in {B}(\mathcal{D},r)$:  
\begin{align}
\label{equation_to_derive_for_y}
h(\mathcal{D}',\mathbf{x})=l \Longleftrightarrow p_l'  > \max_{j\neq l}p_j'.
\end{align}
However, it is challenging to compute $p_l'$ and $\max_{j\neq l}p_j'$ due to the complicated base learning algorithm $\mathcal{A}$. To address the challenge,  we aim to derive a lower bound of $p_l'$ and an upper bound of $\max_{j\neq l}p_j'$, where the lower bound and upper bound are independent from the base learning algorithm $\mathcal{A}$ and  can be easily computed for a given $r$.  
In particular,  we  derive the lower bound and upper bound as the probabilities that the random variable $Y$ is in certain regions of the space $\Omega$ via the Neyman-Pearson Lemma~\citep{neyman1933ix}. Then, we can find the maximal $r$ such that the lower bound is larger than the upper bound for any $\mathcal{D}^{\prime} \in B(\mathcal{D},r)$, and such maximal $r$ is our certified poisoning size $r^{*}$. 


Next, we show the high-level idea of our approach to derive the lower and upper bounds (details are in Supplemental Material). 
Our key idea is to construct regions in the space $\Omega$ such that the random variables $X$ and $Y$ satisfy the conditions of the Neyman-Pearson Lemma~\citep{neyman1933ix}, which enables us to derive the lower and upper bounds using the probabilities that $Y$ is in these regions. Next, we discuss how to construct the regions. Suppose we have a lower bound $\underline{p_l}$ of the largest label probability ${p_l}$ and an upper bound $\overline{p}_s$ of the second largest label probability $p_s$ when the ensemble classifier is trained on the clean training dataset. Formally, 
$\underline{p_l}$ and $\overline{p}_s$ satisfy:
\begin{align}
{\small
\label{equation_of_condition_probability_bagging}
p_l \geq \underline{p_l} \geq \overline{p}_s  \geq p_s= \max_{j\neq l} p_j. 
}
\end{align}
We use the probability bounds instead of the exact label probabilities ${p_l}$ and ${p}_s$, because it is challenging to  compute them exactly.
We first divide the space $\Omega$ into three regions $\mathcal{B}$, $\mathcal{C}$, and $\mathcal{E}$, which include subsamples with $k$ examples sampled from $\mathcal{D}$, $\mathcal{D}^{\prime}$, and $\mathcal{D}\cap \mathcal{D}^{\prime}$, respectively. Then, we can find a region $\mathcal{B}^{\prime} \subseteq \mathcal{E}$ such that we have $\text{Pr}(X \in \mathcal{B}\cup \mathcal{B}^{\prime})=\underline{p_l} -\delta_{l}$, where $\delta_l = \underline{p_l} - (\lfloor \underline{p_l}\cdot n^{k} \rfloor)/n^{k}$ is a small residual. We have the residual $\delta_l$ because $\text{Pr}(X \in \mathcal{B}\cup \mathcal{B}^{\prime})$ is an integer multiple of $\frac{1}{n^k}$.  The reason we assume we can find such region $\mathcal{B}^{\prime}$ is that we aim to derive a sufficient condition. Similarly, we can find $\mathcal{C}_s \subseteq \mathcal{E}$ such that we have $\text{Pr}(X \in  \mathcal{C}_{s}) = \overline{p}_s +\delta_{s}$, where $\delta_s = (\lceil \overline{p}_{s} \cdot n^{k} \rceil)/n^{k} -\overline{p}_{s}$ is a small residual. Given these regions, we leverage the Neyman-Pearson Lemma~\citep{neyman1933ix} to derive a lower bound of $p_l'$ and an upper bound of $\max_{j\neq l}p_j'$ as follows: 
\begin{align}
    p_l' &\geq \text{Pr}(Y \in  \mathcal{B}\cup \mathcal{B}^{\prime}), \\
   \max_{j\neq l}p_j' &\leq \text{Pr}(Y \in  \mathcal{C}\cup \mathcal{C}_{s}), 
\end{align}
where the lower bound $\text{Pr}(Y \in  \mathcal{B}\cup \mathcal{B}^{\prime})$ and upper bound $\text{Pr}(Y \in  \mathcal{C}\cup \mathcal{C}_{s})$ can be easily computed for a given $r$. 
Finally, we find the maximal $r$ such that the lower bound is still larger than the upper bound, which is our certified poisoning size $r^{*}$. 
The following Theorem~\ref{certified_radius_bagging} formally summarizes our certified robustness guarantee of bagging.

\begin{theorem}[Certified Poisoning Size of Bagging]
\label{certified_radius_bagging}
Given a training dataset $\mathcal{D}$, a deterministic or randomized base learning algorithm $\mathcal{A}$, and a testing example $\mathbf{x}$. The ensemble classifier $h$ in bagging is defined in Equation (\ref{baggingdefinition}). Suppose $l$ and $s$ respectively are the labels with the largest and second largest label probabilities predicted by $h$ for $\mathbf{x}$. Moreover, the probability bounds $\underline{p_l}$ and $\overline{p}_s$ satisfy (\ref{equation_of_condition_probability_bagging}).  
Then, $h$ still predicts label $l$ for $\mathbf{x}$ when the number of poisoned training examples is bounded by $r^*$, i.e., we have: 
\begin{align}
 h(\mathcal{D}', \mathbf{x})=l, \forall \mathcal{D}^{\prime}\in {B}(\mathcal{D},r^*), 
\end{align}
where $r^{*}$ is the solution to the following optimization problem: 
\begin{align}
r^{*} & =  \argmax_{r} r \nonumber \\
& \text{s.t.} \max_{n-r \leq n' \leq n+r}(\frac{n^{\prime}}{n})^{k} - 2 \cdot (\frac{\max(n,n^{\prime})-r}{n})^{k} \nonumber\\
\label{certified_condition_bagging}
& + 1 - (\underline{p_l}-\overline{p}_s -\delta_l - \delta_s) < 0, 
\end{align}
where $n=|\mathcal{D}|$, $n^{\prime}=|\mathcal{D}^{\prime}|$, $\delta_l = \underline{p_l} - (\lfloor \underline{p_l}\cdot n^{k} \rfloor)/n^{k}$, and $\delta_s = (\lceil \overline{p}_{s} \cdot n^{k} \rceil)/n^{k} -\overline{p}_{s}$. 
\end{theorem}

Given Theorem~\ref{certified_radius_bagging}, we have the following corollaries. 

\begin{corollary}
\label{corollary1}
Suppose a data poisoning attack only modifies existing training examples. Then, we have $n'=n$ and the solution to optimization problem~(\ref{certified_condition_bagging}) is $r^{*}=\lceil n\cdot (1 - \sqrt[k]{1-\frac{\underline{p_l}-\overline{p}_s - \delta_l - \delta_s}{2}}) -1 \rceil$. 
\end{corollary}

\begin{corollary}
\label{corollary2}
Suppose a data poisoning attack only deletes existing training examples.
Then, we have $n'= n - r$ and  $r^{*}=\lceil n\cdot (1 - \sqrt[k]{1-(\underline{p_l}-\overline{p}_s - \delta_l - \delta_s)}) -1 \rceil$. 
\end{corollary}

\begin{corollary}
\label{corollary3}
Suppose a data poisoning attack only inserts new training examples. 
Then, we have $n' = n + r$ and $r^{*}=\lceil n\cdot ( \sqrt[k]{1+(\underline{p_l}-\overline{p}_s - \delta_l - \delta_s)} - 1) -1 \rceil$. 
\end{corollary}

The next theorem shows that our derived certified poisoning size is tight. 

\begin{theorem}[Tightness of the Certified Poisoning Size]
\label{tightness_theorem}
Assuming we have $\underline{p_l} + \overline{p}_s  \leq 1$, $\underline{p_l}  +(c-1)\cdot \overline{p}_s \geq 1$, and $ \delta_l = \delta_s =0$. Then, for any $r>r^{*}$, there exist a base learning algorithm $\mathcal{A}^{*}$ consistent with~(\ref{equation_of_condition_probability_bagging}) and a poisoned training dataset  $\mathcal{D}^{\prime}$ with $r$ poisoned training examples such that $h(\mathcal{D}',\mathbf{x})\neq l$
or there exist ties. 
\end{theorem}

We have several remarks about our theorems. 

\myparatight{Remark 1} Our Theorem~\ref{certified_radius_bagging} is applicable for any base learning algorithm $\mathcal{A}$, i.e., bagging with any base learning algorithm is provably robust against data poisoning attacks.

\myparatight{Remark 2} For any lower bound $\underline{p_l}$  of the largest label probability and upper bound $\overline{p}_s$ of the second largest label probability, Theorem~\ref{certified_radius_bagging} derives a certified poisoning size. Moreover, our certified poisoning size is related to the gap between the two probability bounds.  If we can estimate tighter probability bounds, then the certified poisoning size may be larger. 

\myparatight{Remark 3} Theorem~\ref{tightness_theorem} shows that when no assumptions on the base learning algorithm are made, it is impossible to certify a poisoning size that is larger than ours.

%% file: practice.tex
\section{Computing the Certified Poisoning Size}

Given a base learning algorithm $\mathcal{A}$, a training dataset $\mathcal{D}$, subsampling size $k$, and $e$ testing examples in $\mathcal{D}_e$, we aim to compute the label ${l}_i$ predicted by the ensemble classifier and the corresponding certified poisoning size ${r}_i^*$ for each testing example $\mathbf{x}_i$. For a testing example $\mathbf{x}_i$, our certified poisoning size relies on a lower bound $\underline{p_{{l}_i}}$ of the largest label probability and an upper bound $\overline{p}_{{s}_i}$ of the second largest label probability. We design a Monte-Carlo algorithm to estimate the probability bounds for the $e$ testing examples simultaneously via training $N$ base classifiers.  
Next, we first describe estimating the probability bounds. Then, we describe our efficient algorithm to solve the optimization problem in~(\ref{certified_condition_bagging}) with the estimated probability bounds to compute the certified poisoning sizes.

\myparatight{Computing the predicted label and probability bounds for one testing example} We first discuss estimating the predicted label ${l}_i$ and probability bounds $\underline{p_{{l}_i}}$ and $\overline{p}_{{s}_i}$ for one testing example $\mathbf{x}_i$. We first randomly sample $N$ subsamples $\mathcal{L}_1,\mathcal{L}_2,\cdots,\mathcal{L}_{N}$ from $\mathcal{D}$ with replacement, each of which has $k$ training examples. Then, we train a base classifier $f_o$ for each subsample $\mathcal{L}_o$ using the base learning algorithm $\mathcal{A}$, where $o=1,2,\cdots,N$. We use the base classifiers to predict labels for $\mathbf{x}_i$, and we denote by $N_j$ the frequency of label $j$, i.e., $N_j$ is the number of base classifiers that predict label $j$ for $\mathbf{x}_i$. We estimate the label with the largest frequency as the label $l_i$ predicted by the ensemble classifier $h$ for $\mathbf{x}_i$. Moreover, based on the definition of label probability, the frequency $N_{j}$ of the label $j$ among the $N$ base classifiers follows a binomial distribution with parameters $N$ and $p_j$. Therefore, given the label frequencies, we can use the Clopper-Pearson~\citep{clopper1934use}  based method called {SimuEM}~\citep{jia2020certified} to estimate the following probability bounds simultaneously:  
{\small
\begin{align}
\label{estimate1}
    \underline{p_{l_i}}&=\text{Beta}(\frac{\alpha}{c};N_{l_i},N-N_{l_i}+1) \\
    \label{estimate2}
   \overline{p}_j &= \text{Beta}(1-\frac{\alpha}{c};N_j,N-N_j+1), \forall j \neq l_i,
\end{align}
}
where $1 - \alpha$ is the confidence level and $\text{Beta}(\beta;\lambda,\theta)$ is the $\beta$th quantile of the Beta distribution with shape parameters $\lambda$ and $\theta$. One natural method  to estimate $\overline{p}_{s_i}$ is that $\overline{p}_{s_i} = \max_{j\neq l_i}\overline{p}_j$. However, this bound may be loose. For example, $\underline{p_{l_i}}+\overline{p}_{s_i}$ may be larger than 1. Therefore, we estimate $\overline{p}_{s_i}$ as $\overline{p}_{s_i} = \min (\max_{j\neq l_i}\overline{p}_j, 1 - \underline{p_{l_i}})$.

\myparatight{Computing the predicted labels and probability bounds for $e$ testing examples} One way of estimating the predicted labels and probability bounds for $e$ testing examples is to apply the above process for each testing example separately. However, such process requires training $N$ base classifiers for each testing example, which is computationally intractable. To address the challenge, we propose a method to estimate them for $e$ testing examples simultaneously via training $N$ base classifiers in total.  Our key idea is to divide the confidence level among the $e$ testing examples such that we can estimate their predicted labels and probability bounds using the same $N$ base classifiers with a simultaneous confidence level at least $1-\alpha$. Specifically, we still use the $N$ base classifiers to predict the label for each testing example as we described above. Then, we follow the above process to estimate the probability bounds  $\underline{p_{{l}_i}}$ and $\overline{p}_{{s}_i}$ for a testing example $\mathbf{x}_i$ via replacing $\alpha$ as $\alpha/e$ in Equation (\ref{estimate1}) and (\ref{estimate2}). Based on the \emph{Bonferroni correction}, the simultaneous confidence level of estimating the probability bounds for the $e$ testing examples is at least $1-\alpha$.

\myparatight{Computing the certified poisoning sizes} Given the estimated probability bounds $\underline{p_{{l}_i}}$ and $\overline{p}_{{s}_i}$ for a testing example $\mathbf{x}_i$, we solve the optimization problem in~(\ref{certified_condition_bagging}) to obtain its certified poisoning size $r_{i}^*$. We design an efficient binary search based method to solve $r_{i}^*$. 
Specifically, we use binary search to find the largest $r$ such that the constraint in (\ref{certified_condition_bagging}) is satisfied. We denote the left-hand side of the constraint as $\max_{n-r \leq n' \leq n+r} L(n)$. For a given $r$, a naive way to check whether the constraint $\max_{n-r \leq n' \leq n+r} L(n')<0$ holds is to check whether $L(n')<0$ holds for each $n'$ in the range $[n-r, n+r]$, which could be inefficient when $r$ is large. To reduce the computation cost, we derive an analytical form of $n'$ at which $L(n')$ reaches its maximum value. Our analytical form enables us to only check whether $L(n')<0$ holds for at most two different $n'$ for a given $r$. The  details of deriving  the analytical form are shown in Supplemental Material. 

\begin{algorithm}[tb]
   \caption{\textsc{Certify}}
   \label{alg:certify}
\begin{algorithmic}
   \STATE {\bfseries Input:} $\mathcal{A}$, $\mathcal{D}$, $k$, $N$, $\mathcal{D}_e$, $\alpha$.
   \STATE {\bfseries Output:} Predicted label and certified poisoning size for each testing example. \\
   $f_1,f_2,\cdots,f_N \gets  \textsc{TrainUnderSample}(\mathcal{A},\mathcal{D},k,N)$ \\
   \FOR{$\mathbf{x}_i$ {\bfseries in} $\mathcal{D}_e$}
   \STATE counts$[j] \gets \sum_{o=1}^{N}\mathbb{I}(f_{o}(\mathbf{x}_i)=j), j\in \{1,2,\cdots,c\} $ \\
   \STATE $l_i,s_i \gets$ top two indices in counts (ties are broken uniformly at random). \\
   \STATE $\underline{p_{l_i}}, \overline{p}_{s_i} \gets \textsc{SimuEM}(\text{counts},\frac{\alpha}{e})$ \\
   \IF{$\underline{p_{l_i}}> \overline{p}_{s_i}$}
   \STATE $r_i^* \gets \textsc{BinarySearch} (\underline{p_{l_i}}, \overline{p}_{s_i}, k, |\mathcal{D}|)$ 
   \ELSE
   \STATE $l_i,r_i^* \gets \text{ABSTAIN},\text{ABSTAIN}$ 
   \ENDIF
   \ENDFOR
  \STATE \textbf{return} $l_1,l_2,\cdots, l_e$ and $r_1^*,r_2^*,\cdots, r_e^*$
\end{algorithmic}
\end{algorithm}

\myparatight{Complete certification algorithm}
Algorithm~\ref{alg:certify} shows our certification process to estimate the predicted labels and certified poisoning sizes for $e$ testing examples in $\mathcal{D}_e$. The function \textsc{TrainUnderSample} randomly samples $N$ subsamples and trains $N$ base classifiers. The function \textsc{SimuEM} estimates the probability bounds $\underline{p_{l_i}}$ and $\overline{p}_{s_i}$ with confidence level $1-\frac{\alpha}{e}$. The function \textsc{BinarySearch} solves the optimization problem in (\ref{certified_condition_bagging}) using the estimated probability bounds $\underline{p_{l_i}}$ and $\overline{p}_{s_i}$ to obtain the certified poisoning size $r_i^*$ for testing example $\mathbf{x}_i$. 

Since the probability bounds are estimated using a Monte Carlo algorithm, they may be estimated incorrectly, i.e., $\underline{p_{l_i}}> p_{l_i}$ or  $\overline{p}_{s_i} < {p}_{s_i}$. When they are estimated incorrectly, our algorithm \textsc{Certify} may output an incorrect certified poisoning size. However,  the following theorem shows that 
the probability that \textsc{Certify} returns an incorrect certified poisoning size for at least one testing example is at most $\alpha$. 
\begin{theorem}
\label{proposition_of_certify}
The probability that $\textsc{Certify}$ returns an incorrect certified poisoning size for at least one testing example in $\mathcal{D}_e$ is at most $\alpha$, i.e., we have: 
{\small
\begin{align}
     &\text{Pr}(\cap_{\mathbf{x}_i \in \mathcal{D}_e} ((\forall \mathcal{D}'\in B(\mathcal{D}, r_i^*), h(\mathcal{D}',\mathbf{x}_i)=l_i)|l_i \neq \text{ABSTAIN} ))    \geq 1 -\alpha.
\end{align}
}
\end{theorem}

%% file: evaluation.tex
\section{Experiments}
\subsection{Experimental Setup}
\noindent
{\bf Datasets and classifiers:} We use MNIST and CIFAR10 datasets. The base learning algorithm is neural network, and we use the example convolutional neural network architecture
and ResNet20~\citep{he2016deep}
in Keras for MNIST and CIFAR10, respectively. The number of training examples in the two datasets are $60,000$ and $50,000$, respectively, which are the training datasets that we aim to certify. Both datasets have 10,000 testing examples, which are the $\mathcal{D}_e$ in our algorithm. When we train a base classifier, we adopt the example data augmentation in Keras for both datasets.

\myparatight{Evaluation metric} We use  \emph{certified accuracy} as our evaluation metric. In particular, we define the certified accuracy at $r$ poisoned training examples of a classifier as the fraction of testing examples whose labels are correctly predicted by the classifier and whose certified poisoning sizes are at least $r$. Formally, we have  the certified accuracy $CA_r$ at $r$ poisoned training examples as follows:
\begin{align}
    CA_r=\frac{\sum_{\mathbf{x}_i \in \mathcal{D}_e} \mathbb{I}(l_i=y_i)\cdot \mathbb{I}(r_{i}^{*}\geq r)}{|\mathcal{D}_e|},
\end{align} 
where $y_i$ is the ground truth label for testing example $\mathbf{x}_i$, and $l_i$ and $r_{i}^{*}$ respectively are the label predicted by the classifier and the corresponding certified poisoning size for $\mathbf{x}_i$. Intuitively,
$CA_r$ of a classifier means that, when the number of poisoned training examples is $r$, the classifier's testing accuracy for $\mathcal{D}_e$ is at least $CA_r$ no matter how the attacker manipulates the $r$ poisoned training examples. Based on Theorem~\ref{proposition_of_certify}, the $CA_r$ computed using the predicted labels and certified poisoning sizes outputted by our \textsc{Certify} algorithm has a confidence level $1-\alpha$. 

\myparatight{Parameter setting} Our method has three parameters, i.e., $k$, $\alpha$, and $N$. Unless otherwise mentioned, we adopt the following default settings for them:  $\alpha = 0.001$,  $N=1,000$, $k=30$ for MNIST, and $k=500$ for CIFAR10.  We will study the impact of each parameter while setting the remaining parameters to their default values. Note that training the $N$ base classifiers can be easily parallelized.  We performed experiments on a server with 80 CPUs@2.1GHz, 8 GPUs (RTX 6,000), and 385 GB main memory.

\begin{figure*}[!t]
    \center
    {\includegraphics[width=0.24\textwidth]{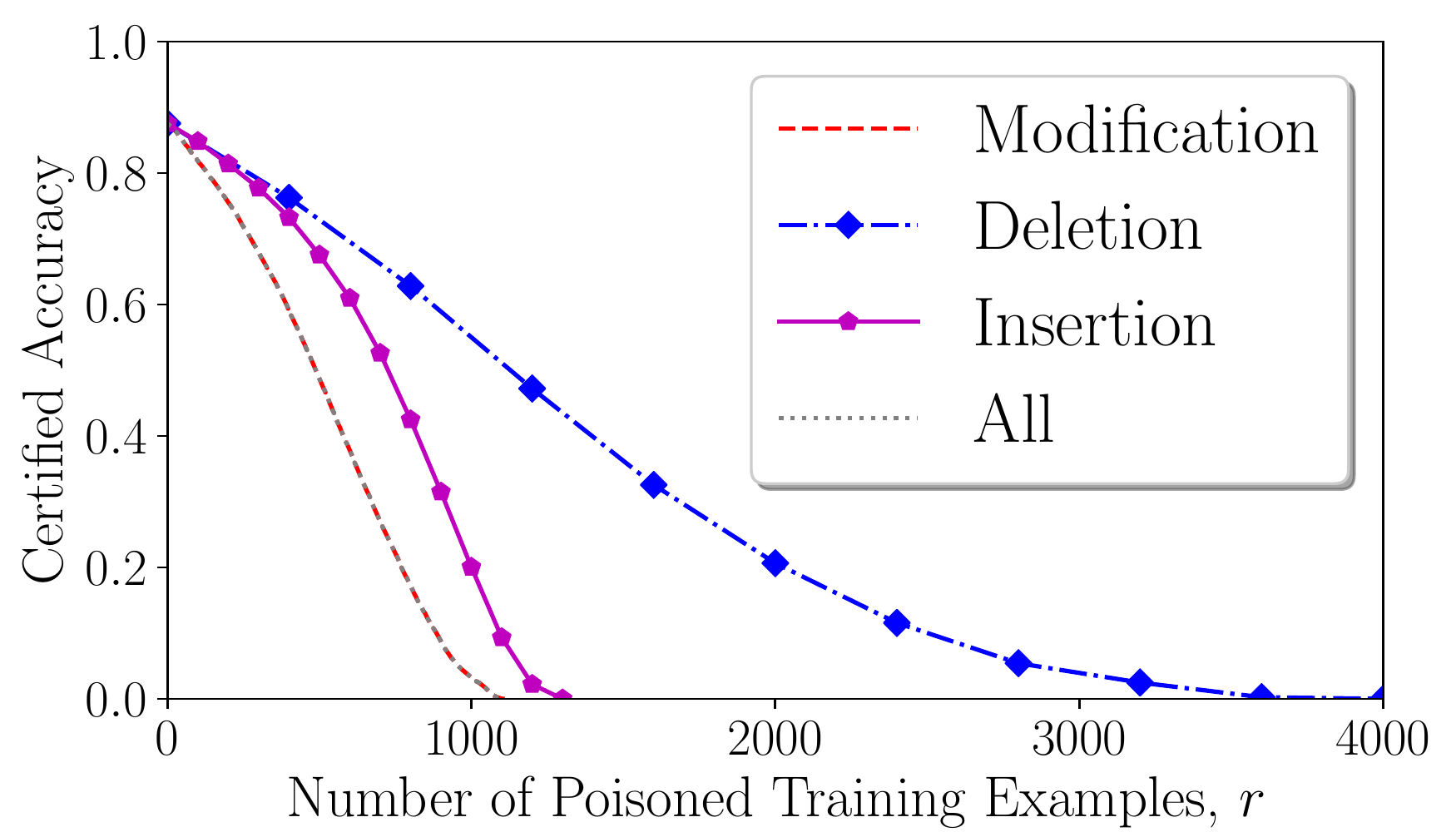}} 
    {\includegraphics[width=0.24\textwidth]{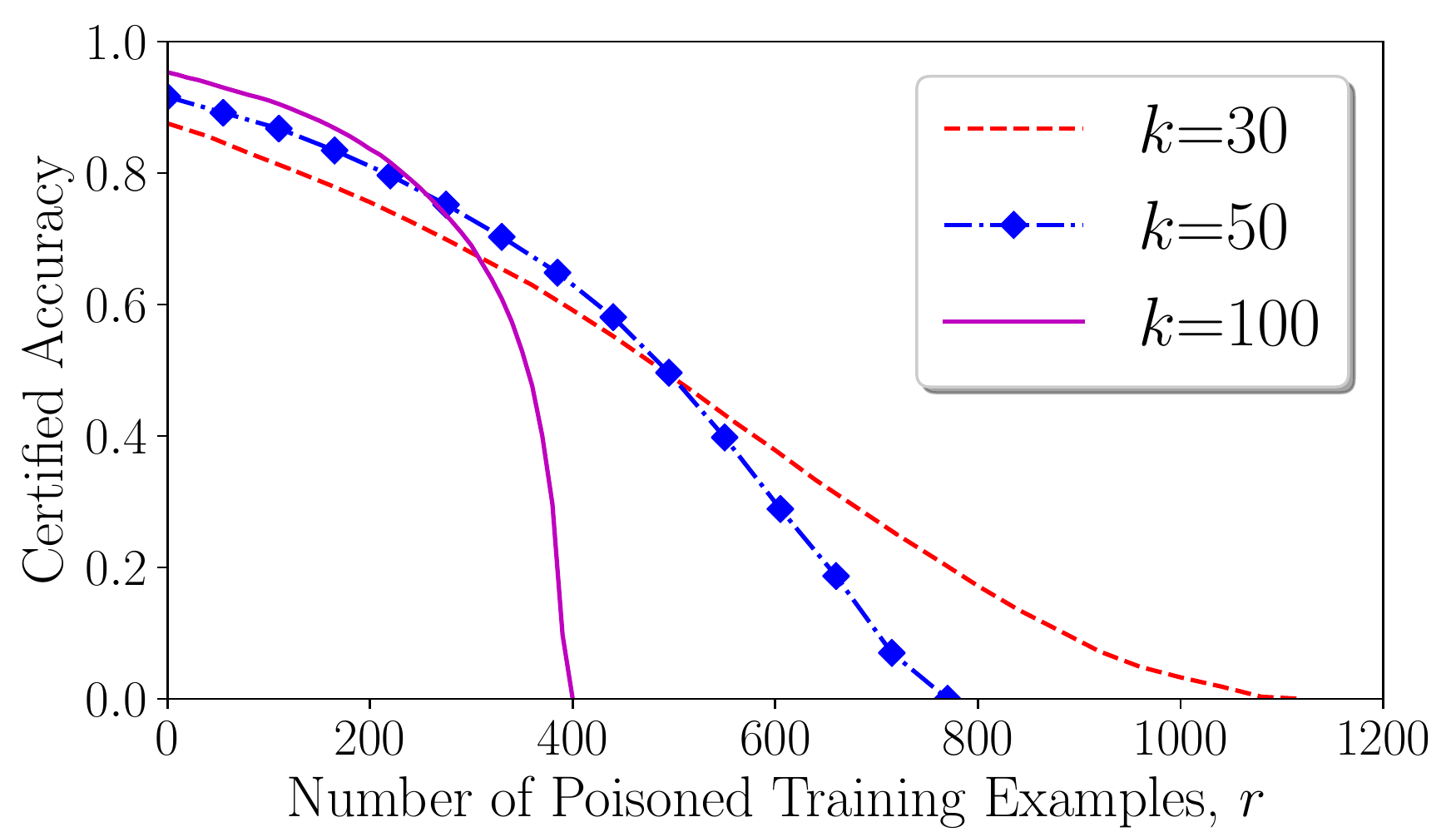}}
    {\includegraphics[width=0.24\textwidth]{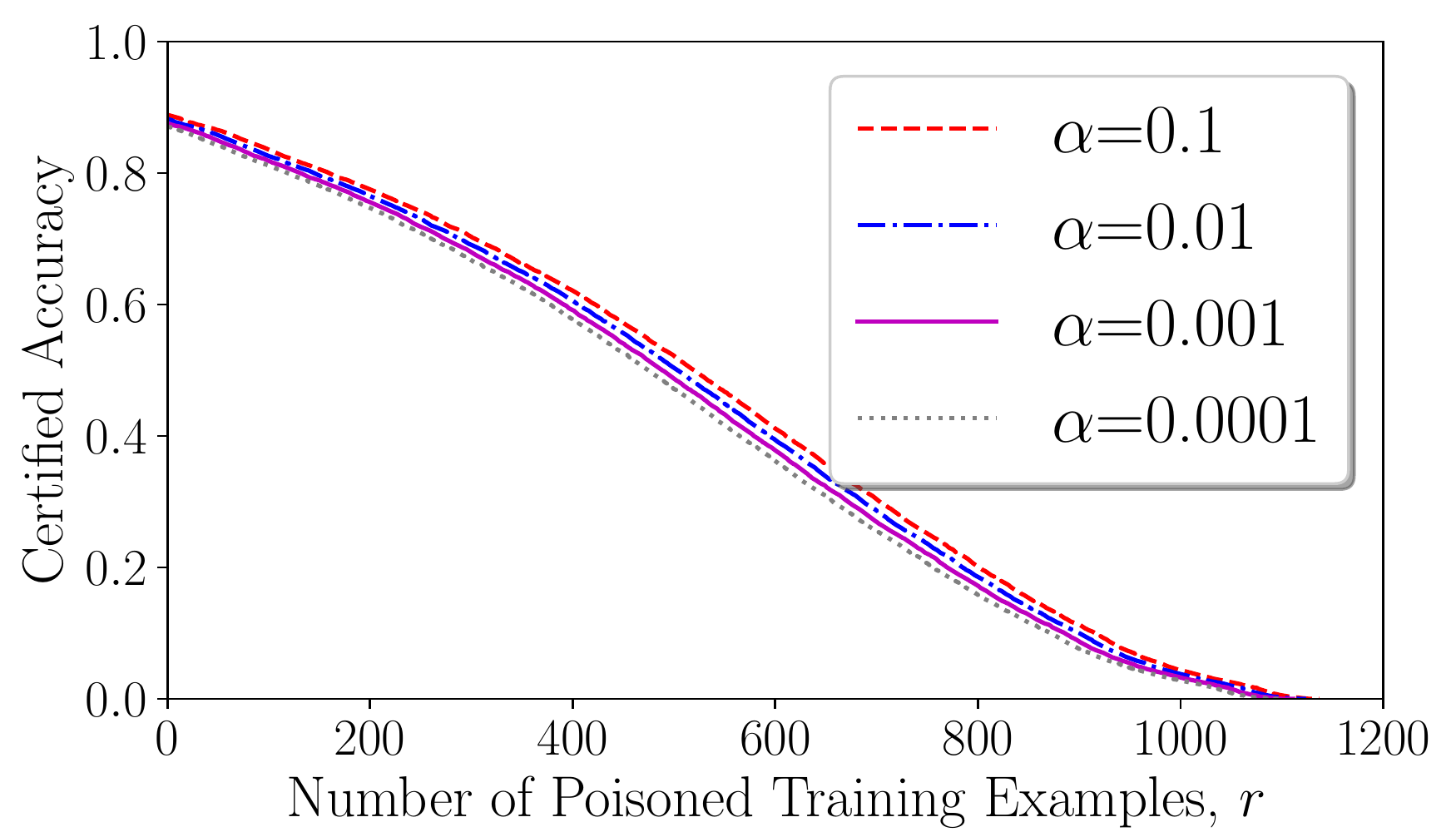}}
    {\includegraphics[width=0.24\textwidth]{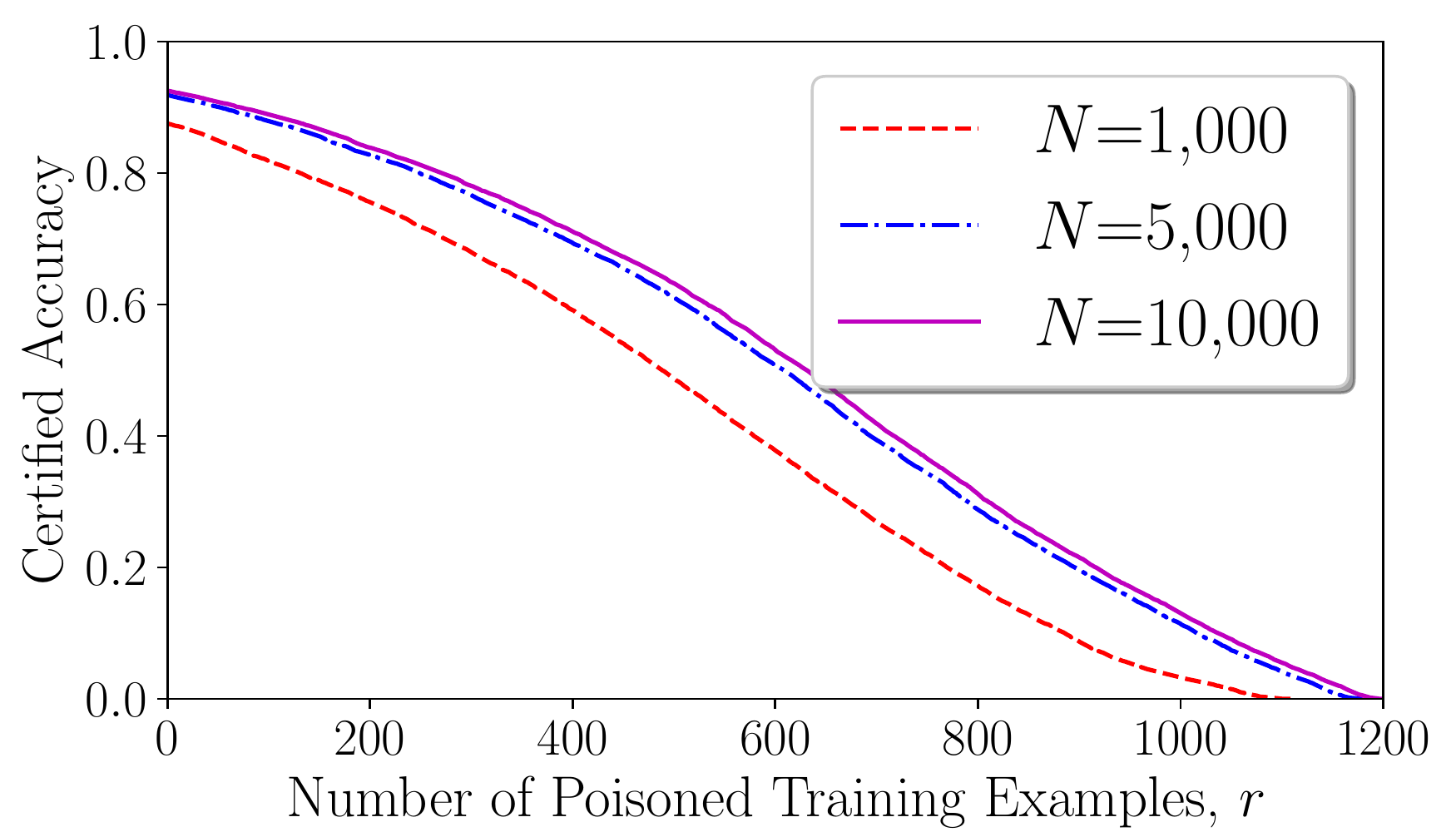}}\\
    \subfigure[Different attacks]{\includegraphics[width=0.24\textwidth]{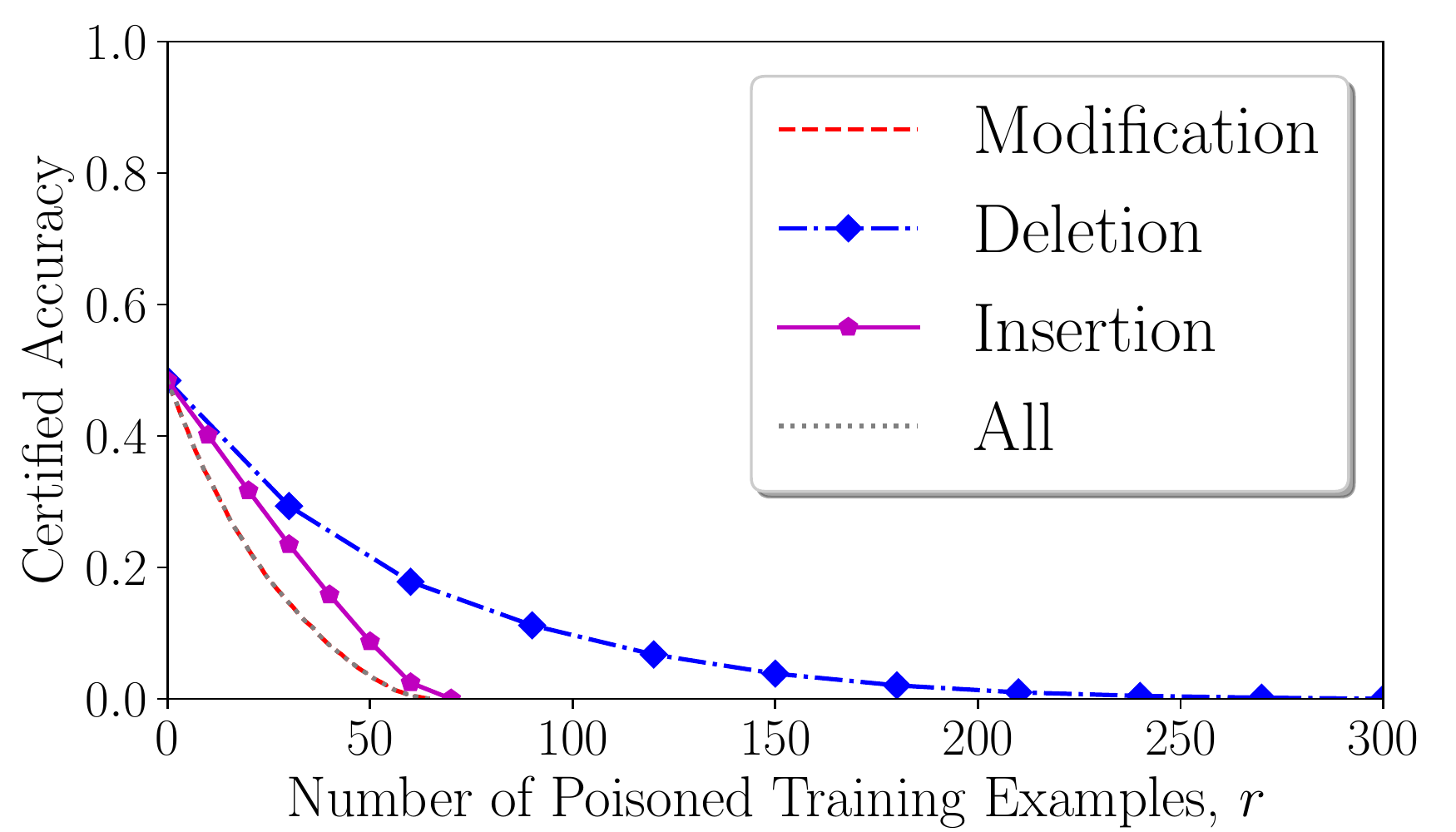}\label{fig:vary_case}}
    \subfigure[Impact of $k$]{\includegraphics[width=0.24\textwidth]{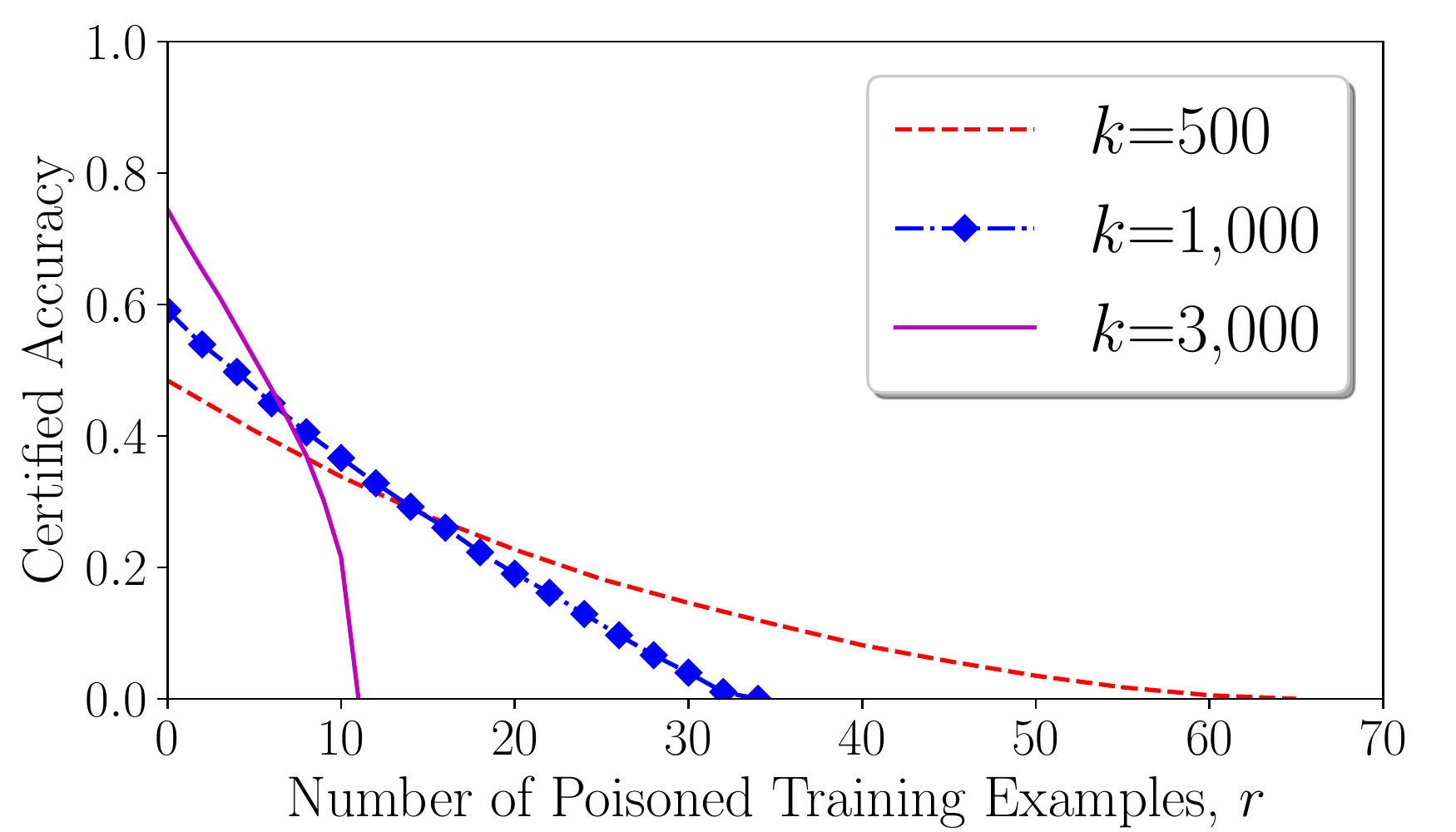}\label{fig:vary_k}}  
    \subfigure[Impact of $\alpha$]{\includegraphics[width=0.24\textwidth]{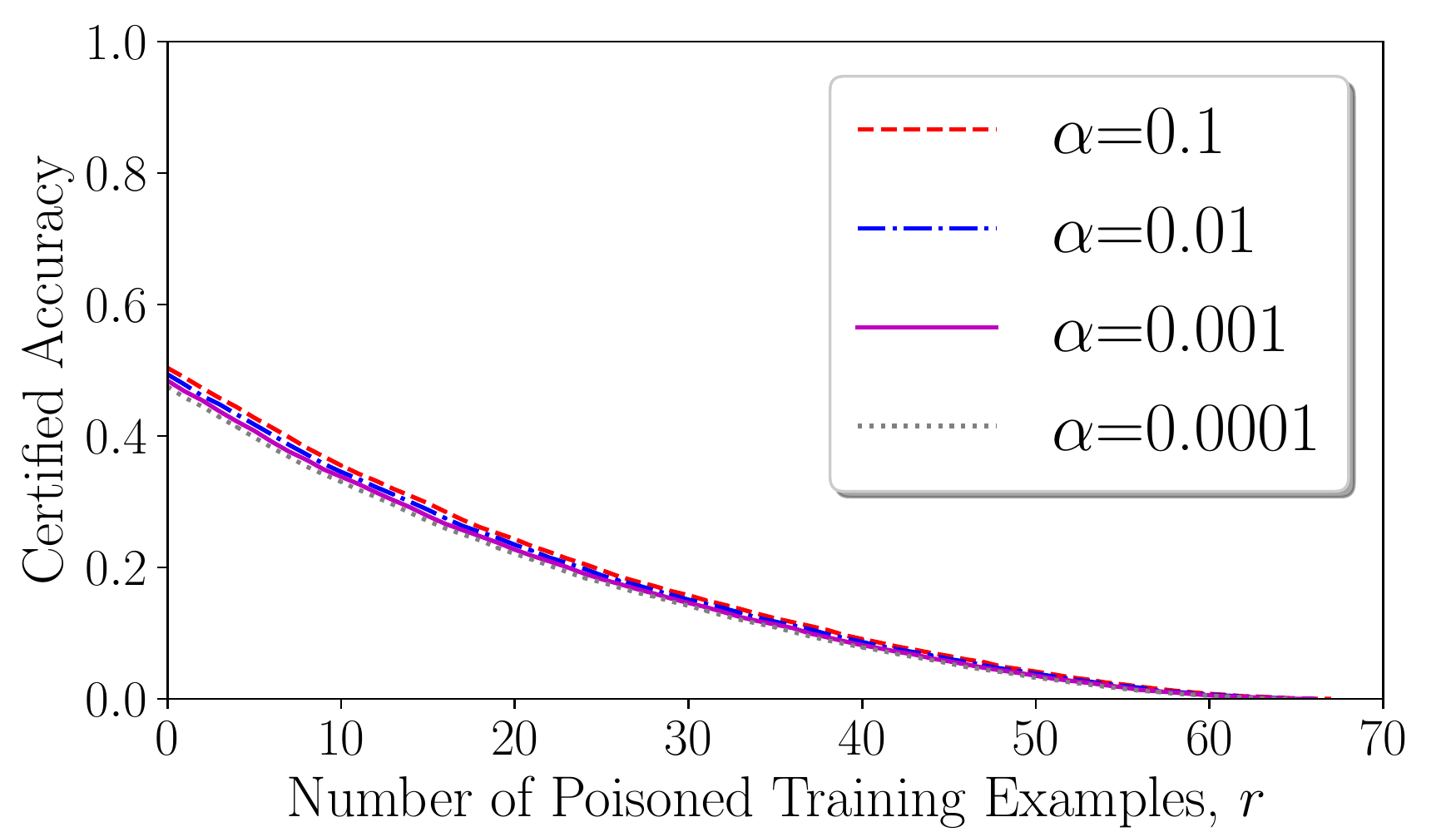}\label{fig:vary_alpha}}
    \subfigure[Impact of $N$]{\includegraphics[width=0.24\textwidth]{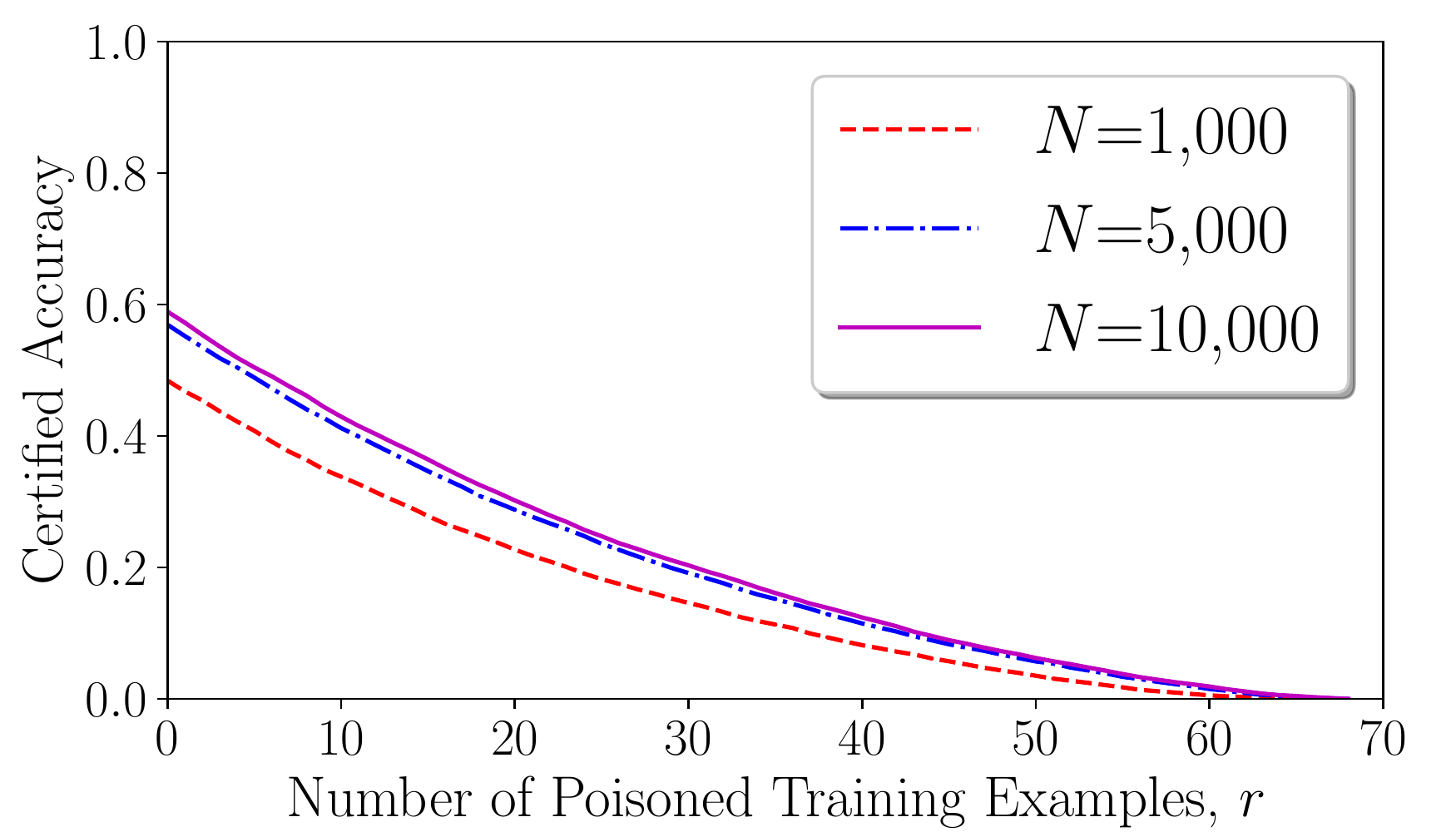}\label{fig:vary_N}}
    \vspace{-3mm}
    \caption{(a) Comparing different data poisoning attacks. (b)-(d) Impact of $k$, $\alpha$, and $N$ on the certified accuracy of our method. The first row is the result on MNIST and the second row is the result on CIFAR10.}
    \label{fig:vary_k_alpha_N}
\end{figure*}

\begin{figure*}[!t]
    \center
    \subfigure[Transfer learning]{\includegraphics[width=0.3\textwidth]{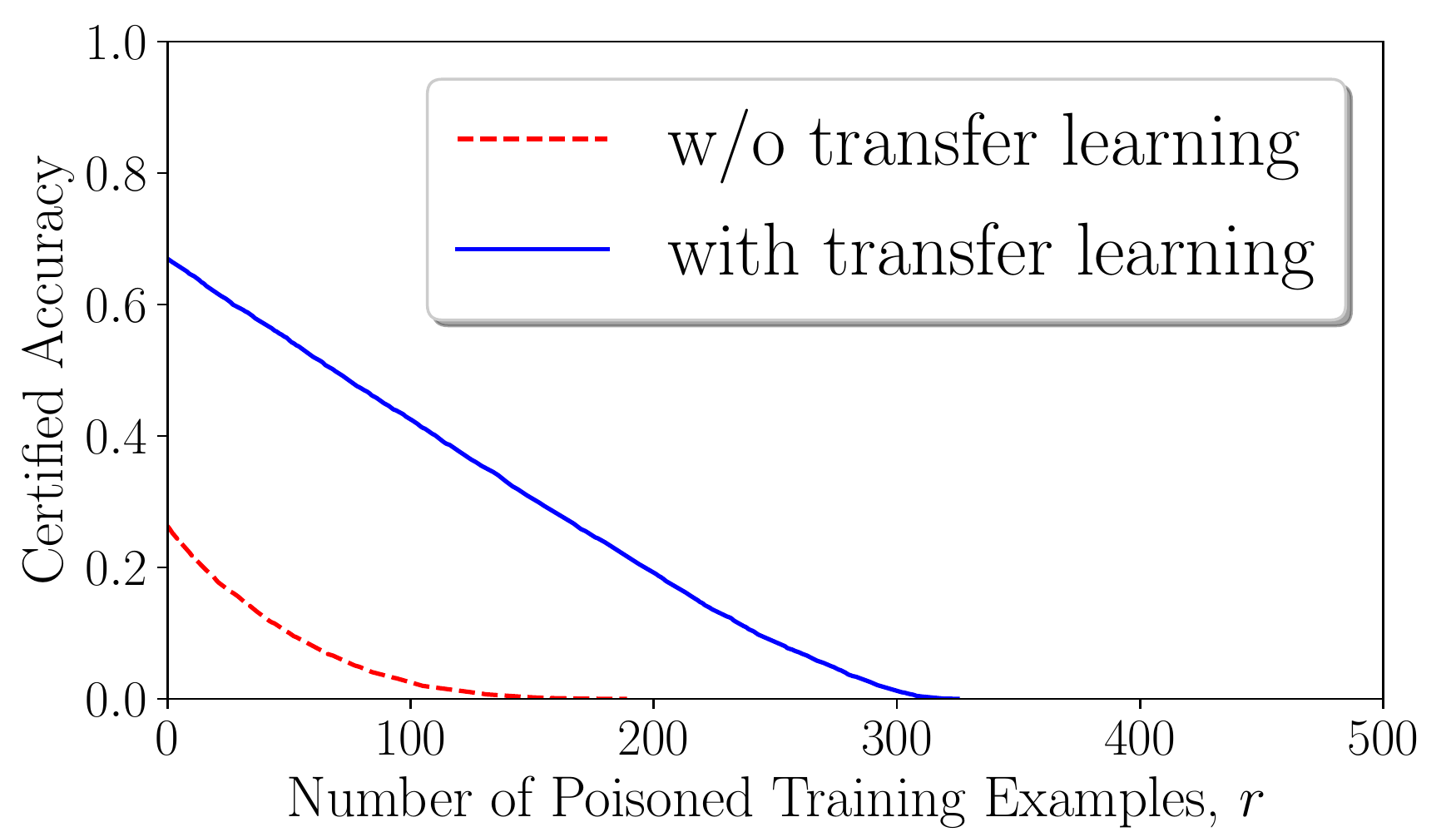}\label{fig:transfer}}  
    \subfigure[Comparing certified accuracy]{\includegraphics[width=0.3\textwidth]{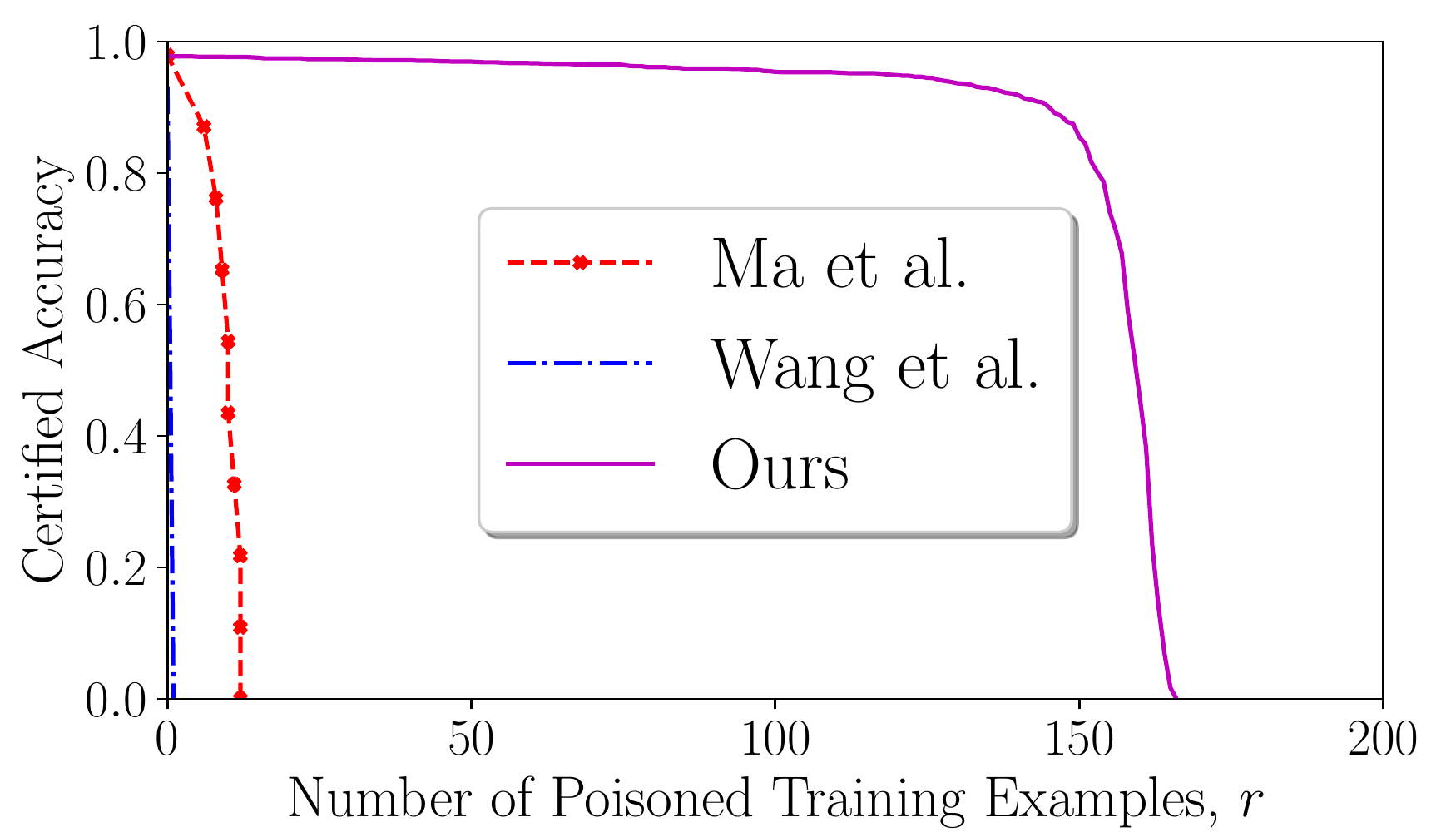}\label{fig:compare_acc}}  
    \subfigure[Comparing running time]{\includegraphics[width=0.3\textwidth]{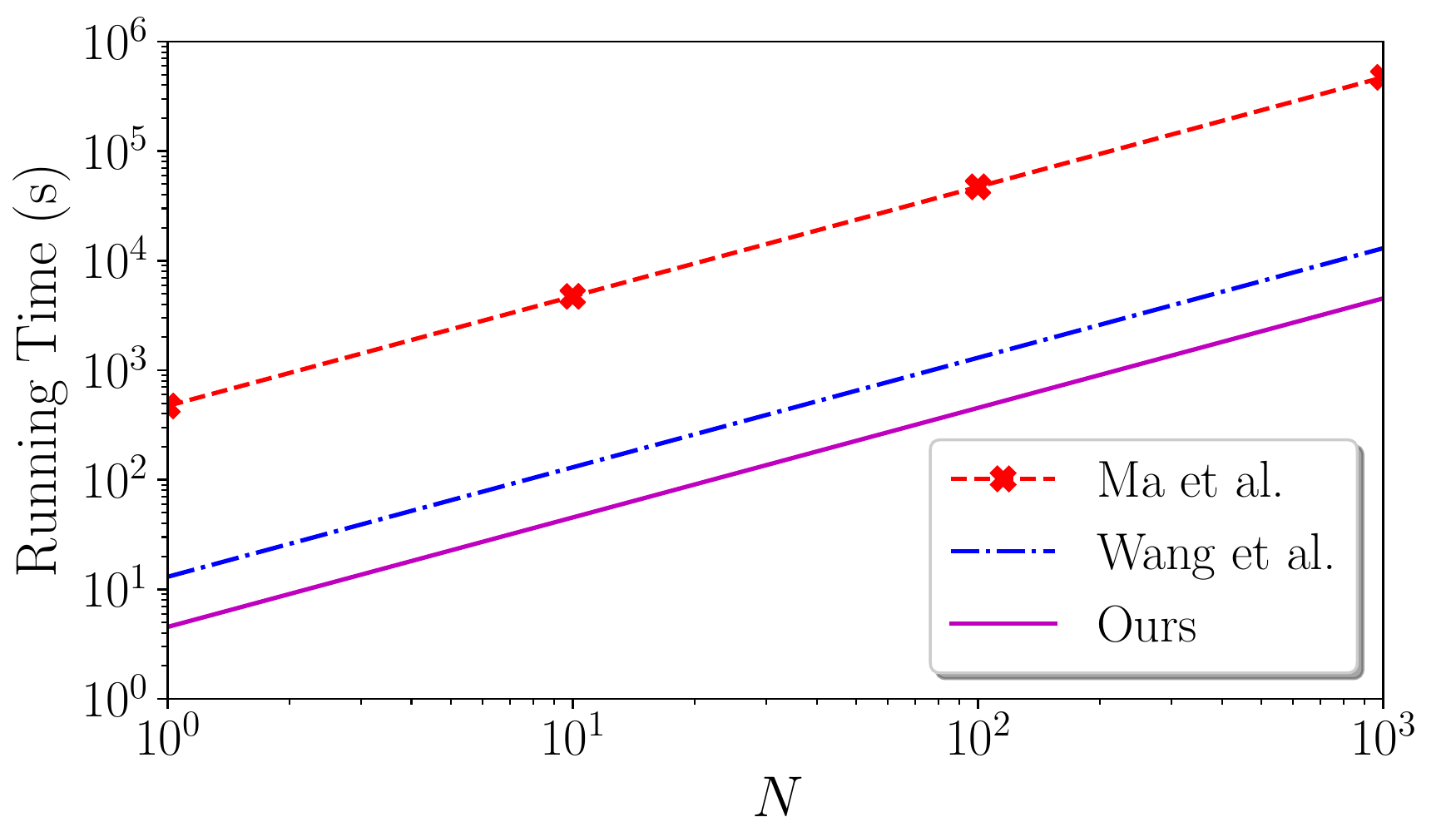}\label{fig:compare_time}} 
    \vspace{-3mm}
    \caption{(a) Transfer learning improves our certified accuracy on CIFAR10. Comparing  our method with existing methods with respect to (b) certified accuracy and (c) running time on the MNIST 1/7 dataset. }
    \label{fig:compare}
\end{figure*}

\subsection{Experimental Results}
\noindent
{\bf Comparing different data poisoning attacks:} An attacker can modify, delete, and/or insert training examples in data poisoning attacks. We compare the certified accuracy of our method when an attacker only modifies, deletes, or inserts training examples. Our Corollary \ref{corollary1}-\ref{corollary3} show the certified poisoning sizes for such attacks.   
Figure~\ref{fig:vary_case} shows the comparison results, where ``All'' corresponds to the attacks that can use modification, deletion, and insertion. Our method achieves the best certified accuracy for attacks that only delete training examples. This is because deletion simply reduces the size of the clean training dataset. The curves corresponding to Modification and All overlap and have the lowest certified accuracy. This is because modifying a training example is equivalent to deleting an existing training example and inserting a new one. In the following experiments, we use the All attacks unless otherwise mentioned. 

\myparatight{Impact of $k$, $\alpha$, and $N$} Figure~\ref{fig:vary_k_alpha_N} shows the impact of $k$,  $\alpha$, and $N$ on the certified accuracy of our method. As the results show,  $k$ controls a tradeoff between accuracy under no poisoning and robustness. Specifically, when $k$ is larger, our method has a higher accuracy when there are no data poisoning attacks (i.e., $r=0$) but the certified accuracy drops more quickly as the number of poisoned training examples increases. The reason is that a larger $k$ makes it more likely to sample poisoned training examples when creating the subsamples in bagging. The certified accuracy increases as $\alpha$ or $N$ increases. The reason is that a larger  $\alpha$ or $N$  produces tighter estimated probability bounds, which make the certified poisoning sizes larger. We also observe  that the certified accuracy is relatively insensitive to $\alpha$. 

\myparatight{Transfer learning improves certified accuracy}
Our method trains multiple base classifiers and each base classifier is trained using $k$ training examples. Improving the accuracy of each base classifier can improve the certified accuracy. We explore using transfer learning to train more accurate base classifiers. Specifically, we use the Inception-v3 classifier pretrained on ImageNet to extract features and we use a public implementation\footnote{https://github.com/alexisbcook/keras\_transfer\_cifar10} 
to train our base classifiers on CIFAR10. Figure~\ref{fig:transfer} shows that transfer learning can significantly increase our certified accuracy, where $k=100$, $\alpha=0.001$, and $N=1,000$. Note that we assume the pretrained classifier is not poisoned in this experiment.  

\myparatight{Comparing with \cite{ma2019data}, \cite{wang2020certifying}, and \cite{rosenfeld2020certified}}
 Since these methods are not scalable because they train $N$ classifiers on the entire training dataset, we perform comparisons on the MNIST 1/7 dataset that just includes digits 1 and 7. This subset includes 13,007 training examples and 2,163 testing examples. Note that our above experiments used the entire MNIST dataset.


\begin{itemize}
    \item {\bf \cite{ma2019data}.} Ma et al.  showed that a classifier trained with differential privacy  achieves certified robustness against data poisoning attacks. Suppose $ACC_r$ is the testing accuracy for $\mathcal{D}_e$ of a differentially private classifier trained on a poisoned training dataset with $r$ poisoned training examples. Based on the Theorem 3 in~\citep{ma2019data},  
    we have the expected testing accuracy $E(ACC_r)$ is lower bounded by a certain function of $E(ACC)$, $r$, and $(\epsilon, \delta)$ (the function can be found in their Theorem 3), where $E(ACC)$ is the expected testing accuracy of a differentially private classifier that is trained using the clean training dataset and $(\epsilon, \delta)$ are the differential privacy parameters. The randomness in $E(ACC_r)$ and $E(ACC)$ are from differential privacy. This lower bound is the certified accuracy that the method achieves.  A lower bound of $E(ACC)$ can be further estimated with confidence level $1-\alpha$ via training $N$ differentially private classifiers on the entire clean training dataset. For simplicity, we estimate $E(ACC)$ as the average testing accuracies of the $N$ differentially private classifiers, which gives advantages for this method.    We use DP-SGD~\citep{Abadi16} implemented in TensorFlow 
    to train differentially private classifiers. Moreover, we set $\epsilon=0.3$ and $\delta=10^{-5}$ such that this method and our method achieve comparable certified accuracies when $r=0$.

    \item {\bf \cite{wang2020certifying} and \cite{rosenfeld2020certified}.} Wang et al. proposed a randomized smoothing based method to certify robustness against backdoor attacks via randomly flipping features and labels of training examples as well as features of testing examples. 
    Rosenfeld et al. leveraged randomized smoothing to certify robustness against label flipping attacks. 
    Both methods can be generalized to certify robustness against data poisoning attacks that modify both features and labels of existing training examples via  randomly flipping features and labels of training examples. Moreover, the two methods become the same after such generalization. 
    Therefore, we only show results for \cite{wang2020certifying}. 
    In particular, we binarize the features to apply this method. We train $N$ classifiers to estimate the certified accuracy with a confidence level $1-\alpha$. Unlike our method, when training a classifier, they flip each feature/label value in the training dataset with probability $\beta$ and use the entire noisy training dataset. When predicting the label of a testing example, this method takes a majority vote among the $N$ classifiers. We set $\beta=0.3$ such that this method and our method achieve comparable certified accuracies when $r=0$. We note that this method certifies the number of poisoned features/labels in the training dataset. We transform this certificate to the number of poisoned training examples as $\lfloor \frac{F}{d+1}\rfloor$, where $F$ is the certified number of features/labels and $d+1$ is the number of features/label of a training example ($d$ features + one label). We have $d=784$ for MNIST. 

\end{itemize}

Figure~\ref{fig:compare_acc} shows the comparison results, where $k=50$, $\alpha = 0.001$, and $N=1,000$. To be consistent with previous work, we did not use data augmentation when training the base classifiers for all three methods in these experiments. Our method significantly outperforms existing methods. For example, our method can achieve $96.95\%$ certified accuracy when the number of poisoned training examples is $r=50$, while the certified accuracy is $0$ under the same setting for existing methods. Figure~\ref{fig:compare_time} shows that our method is also more efficient than existing methods. This is because our method trains base classifiers on a small number of training examples while existing methods train classifiers on the entire training dataset.  
Ma et al. outperforms Wang et al. and Rosenfeld et al. because differential privacy directly certifies robustness against modification/deletion/insertion of training examples while randomized smoothing was designed to certify robustness against modifications of features/labels.

%% file: related.tex
\section{Related Work}

Data poisoning attacks  carefully modify, delete, and/or insert some training examples in the training dataset such that a learnt model makes incorrect predictions for many testing examples indiscriminately (i.e., the learnt model has a large testing error rate) or for some attacker-chosen testing examples. 
 For instance,  data poisoning attacks  have been shown to be effective for Bayes classifiers~\citep{nelson2008exploiting}, SVMs~\citep{biggio2012poisoning}, neural networks~\citep{yang2017generative,munoz2017towards,Suciu18,shafahi2018poison}, linear regression models~\citep{mei2015using,Jagielski18}, PCA~\citep{rubinstein2009antidote}, LASSO~\citep{xiao2015feature}, collaborative filtering~\citep{li2016data,yang2017fake,fang2018poisoning,fang2020influence}, clustering~\citep{biggio2013data,biggio2014poisoning}, graph-based methods~\citep{zugner2018adversarial,wang2019attacking,jia2020certifiedcommunity,zhang2020backdoor}, federated learning~\citep{fang2020local,bhagoji2019analyzing,bagdasaryan2020backdoor}, and others~\citep{mozaffari2014systematic,mei2015security,koh2018stronger,zhu2019transferable}. 
We note that backdoor attacks~\citep{gu2017badnets,liu2017trojaning} also poison the training dataset. However, unlike data poisoning attacks, backdoor attacks also inject perturbation (i.e., a trigger) to testing examples.

One category of defenses~\citep{Cretu08,barreno2010security,Suciu18,Tran18} aim to detect the poisoned training examples based on their negative impact on the error rate of the learnt model. 
Another category of defenses~\citep{Feng14,Jagielski18} aim to design  new loss functions, solving which detects the poisoned training examples and learns a model simultaneously. For instance, \cite{Jagielski18}  proposed to jointly optimize the selection of a subset of training examples with a given size and a model  that minimizes the loss function; and the unselected training examples are treated as poisoned ones. 
\cite{steinhardt2017certified} assumes that a model is trained only using examples in a feasible set and derives an approximate upper bound of the loss function for any data poisoning attacks under these assumptions.  
However, all of these defenses cannot certify that the learnt model predicts the same label for a testing example under data poisoning attacks.

\cite{ma2019data} shows that differentially private models  certify robustness against data poisoning attacks.  
\cite{wang2020certifying} proposes to use randomized smoothing to certify robustness against backdoor attacks, which is also applicable to certify robustness against data poisoning attacks. \cite{rosenfeld2020certified} leverages randomized smoothing to certify robustness against label flipping attacks. 
However, these defenses achieve loose certified robustness guarantees. Moreover, \cite{ma2019data} is only applicable to learning algorithms that can be differentially private, while \cite{wang2020certifying} and \cite{rosenfeld2020certified} are  only applicable to data poisoning attacks that modify existing training examples.  \cite{biggio2011bagging} proposed bagging as an empirical defense  against data poisoning attacks. However, they did not derive the certified robustness of bagging. We note that a concurrent work~\citep{levine2020deep} proposed to certify robustness against data poisoning attacks via partitioning the training dataset using a hash function. However, their results are only applicable to deterministic learning algorithms.


%% file: conclusion.tex
\section{Conclusion}
Data poisoning attacks pose severe security threats to machine learning systems. In this work, we show the intrinsic certified robustness of bagging against data poisoning attacks. Specifically, we show that bagging predicts the same label for a testing example when the number of poisoned training examples is bounded. Moreover, we show that our derived bound is tight if no assumptions on the base learning algorithm are made. We also empirically demonstrate the effectiveness of our method using MNIST and CIFAR10. Our results show that our method achieves much better certified robustness and is more efficient than existing certified defenses. Interesting future work includes: 1) generalizing our method to other types of data, e.g., graphs, and 2) improving our method by leveraging meta-learning.

%% file: appendix.tex
\appendix
\onecolumn
\newpage 

\section{Proof of Theorem~\ref{certified_radius_bagging}}
\label{proof_of_certified_radius_bagging}

We first define some notations that will be used in our proof. Given a training dataset $\mathcal{D}$ and its poisoned version $\mathcal{D}^{\prime}$, we define the following two random variables:
\begin{align}
   & X = g(\mathcal{D}) \\
   & Y = g(\mathcal{D}^{\prime}), 
\end{align}
 where $X$ and $Y$ respectively are two random lists with $k$ examples sampled from $\mathcal{D}$ and $\mathcal{D}^{\prime}$ with replacement uniformly at random. We denote by $\mathcal{I}=\mathcal{D}\cap\mathcal{D}^{\prime}$ the set of training examples that are in both $\mathcal{D}$ and $\mathcal{D}'$. We denote $n=|\mathcal{D}|$, $n'=|\mathcal{D}'|$, and $m=|\mathcal{I}|$, which are the number of training examples in $\mathcal{D}$, $\mathcal{D}^{\prime}$, and $\mathcal{D}\cap \mathcal{D}^{\prime}$, respectively. We use $\Omega$ to denote the joint space of random variables $X$ and $Y$, i.e., each element in $\Omega$ is a list with $k$ examples sampled from $\mathcal{D} $ or $ \mathcal{D}^{\prime}$ with replacement uniformly at random. For convenience, we define operators $\sqsubseteq, \not\sqsubseteq$ as follows: 
\begin{definition}[$\sqsubseteq, \not\sqsubseteq$]
Assuming $\omega\in \Omega$ is a list of $k$ examples and $\mathcal{S}$ is a set of examples, we say $\omega \sqsubseteq \mathcal{S}$ if $\text{ }\forall \mathbf{w} \in \omega, \mathbf{w} \in \mathcal{S}$. We say $\omega \not\sqsubseteq \mathcal{S}$ if $\text{ }\exists \mathbf{w} \in \omega, \mathbf{w} \not\in \mathcal{S}$. 
\end{definition}
For instance, we  have $X \sqsubseteq \mathcal{D}$ and $Y \sqsubseteq \mathcal{D}^{\prime}$. 
Before proving our theorem, we  show a variant of the Neyman-Pearson Lemma~\citep{neyman1933ix} that will be used in our proof. 
\begin{lemma}[Neyman-Pearson Lemma]
\label{lemma_np}
Suppose $X$ and $Y$ are two random variables in the space $\Omega$ with probability distributions $\mu_{x}$ and $\mu_{y}$, respectively. Let $M:\Omega\xrightarrow{} \{0,1\}$ be a random or deterministic function. Then, we have the following:  
\begin{itemize}
\item If $S_1=\{\omega\in \Omega:\mu_{x}(\omega) > t\cdot \mu_{y}(\omega)  \}$ and $S_2=\{\omega\in \Omega:\mu_{x}(\omega) = t\cdot \mu_{y}(\omega)  \}$ for some $t > 0$. Let $S=S_1\cup S_3$, where $S_3 \subseteq S_2$. If we have $\text{Pr}(M(X)=1)\geq \text{Pr}(X\in S)$, then $\text{Pr}(M(Y)=1)\geq \text{Pr}(Y\in S)$.

\item If $S_1=\{\omega\in \Omega:\mu_{x}(\omega) < t\cdot \mu_{y}(\omega)  \}$ and $S_2=\{\omega\in \Omega:\mu_{x}(\omega) = t\cdot \mu_{y}(\omega)  \}$ for some $t > 0$. Let $S=S_1\cup S_3$, where $S_3 \subseteq S_2$. If we have $\text{Pr}(M(X)=1)\leq \text{Pr}(X\in S)$, then $\text{Pr}(M(Y)=1)\leq \text{Pr}(Y\in S)$.
\end{itemize}
\end{lemma}
\begin{proof}
We show the proof of the first part, and the second part can be proved similarly. For simplicity, we use $M(1|\omega)$ and $M(0|\omega)$ to denote the probabilities that $M(\omega)=0$ and $M(\omega)=1$, respectively. We use $S^c$ to denote the complement of $S$, i.e., $S^{c} = \Omega\setminus S$. We have the following: 
\begin{align}
 &   \text{Pr}(M(Y)=1) - \text{Pr}(Y\in S) \\
 =& \sum_{\omega\in\Omega}M(1|\omega)\cdot\mu_{y}(\omega)- \sum_{\omega\in S}\mu_{y}(\omega)\\
 =& \sum_{\omega\in S^c}M(1|\omega)\cdot\mu_{y}(\omega)+ \sum_{\omega\in S}M(1|\omega)\cdot\mu_{y}(\omega)-  \sum_{\omega\in S}M(1|\omega)\cdot\mu_{y}(\omega)-\sum_{\omega\in S}M(0|\omega)\cdot\mu_{y}(\omega)\\
 \label{np_lemma_proof_key_1}
 =& \sum_{\omega\in S^c}M(1|\omega)\cdot\mu_{y}(\omega)-\sum_{\omega\in S}M(0|\omega)\cdot\mu_{y}(\omega)\\ 
 \geq & \frac{1}{t} \cdot ( \sum_{\omega\in S^c}M(1|\omega)\cdot\mu_{x}(\omega)-\sum_{\omega\in S}M(0|\omega)\cdot\mu_{x}(\omega) )\\ 
  \label{np_lemma_proof_key_2}
 = & \frac{1}{t} \cdot ( \sum_{\omega\in S^c}M(1|\omega)\cdot\mu_{x}(\omega)+ \sum_{\omega\in S}M(1|\omega)\cdot\mu_{x}(\omega)-  \sum_{\omega\in S}M(1|\omega)\cdot\mu_{x}(\omega)-\sum_{\omega\in S}M(0|\omega)\cdot\mu_{x}(\omega) ) \\
 = &  \frac{1}{t} \cdot (\sum_{\omega\in\Omega}M(1|\omega)\cdot\mu_{x}(\omega)- \sum_{\omega\in S}\mu_{x}(\omega) ) \\
 =&  \frac{1}{t} \cdot ( \text{Pr}(M(X)=1) - \text{Pr}(X\in S)  ) \\
 \geq & 0. 
\end{align}
We obtain~(\ref{np_lemma_proof_key_2}) from~(\ref{np_lemma_proof_key_1}) because $\mu_{x}(\omega) \geq t\cdot \mu_{y}(\omega), \forall \omega \in S$ and  $\mu_{x}(\omega) \leq t\cdot \mu_{y}(\omega), \forall \omega \in S^c$. We have the last inequality because $\text{Pr}(M(X)=1)\geq \text{Pr}(X\in S)$. 
\end{proof}

Next, we prove our Theorem~\ref{certified_radius_bagging}. Our goal is to show that $h(\mathcal{D}', \mathbf{x})=l$, i.e., $\text{Pr}(\mathcal{A}(Y, \mathbf{x})=l) > \max_{j\neq l} \text{Pr}(\mathcal{A}(Y, \mathbf{x})=j)$. Our key idea is to derive a lower bound of $\text{Pr}(\mathcal{A}(Y, \mathbf{x})=l)$ and an upper bound of $\max_{j\neq l} \text{Pr}(\mathcal{A}(Y, \mathbf{x})=j)$, where the lower bound and upper bound can be easily computed. We derive the lower bound and upper bound using the Neyman-Pearson Lemma. Then, we derive the certified poisoning size by requiring the lower bound to be larger than the upper bound. Next, we derive the lower bound, the upper bound, and the certified poisoning size.  

\myparatight{Deriving a lower bound of $\text{Pr}(\mathcal{A}(Y, \mathbf{x})=l)$} We first define the following residual:
\begin{align}
    \delta_l =\underline{p_l} - (\lfloor \underline{p_l}\cdot n^{k} \rfloor)/n^{k}. 
\end{align}
We define a binary function $M(\omega)= \mathbb{I}(\mathcal{A}(\omega,\mathbf{x})=l)$ over the space $\Omega$, where $\omega\in \Omega$ and $\mathbb{I}$ is the indicator function. Then, we have $\text{Pr}(\mathcal{A}(Y, \mathbf{x})=l)=\text{Pr}(M(Y)=1)$. Our idea is to construct a subspace for which we can apply the first part of Lemma~\ref{lemma_np} to derive a lower bound of $\text{Pr}(M(Y)=1)$. 
 We first divide the space $\Omega$  into three subspaces as follows: 
\begin{align}
& \mathcal{B}=\{\omega\in\Omega |\omega \sqsubseteq \mathcal{D}, \omega\not\sqsubseteq \mathcal{I}  \}, \\
& \mathcal{C}=\{\omega\in\Omega |\omega \sqsubseteq \mathcal{D}^{\prime}, \omega\not\sqsubseteq \mathcal{I}  \}, \\
& \mathcal{E}=\{\omega\in\Omega | \omega \sqsubseteq \mathcal{I}  \}.
\end{align}
Since we sample $k$ training examples with replacement uniformly at random, we have the following: 
\begin{align}
&\text{Pr}(X = \omega) = 
\begin{cases}
 \frac{1}{{n^k}}, &\text{ if } \omega\in \mathcal{B}  \cup \mathcal{E} \\
 0, &\text{ otherwise}
\end{cases} \\
&\text{Pr}(Y = \omega) = 
\begin{cases}
 \frac{1}{{(n^{\prime})^k}}, &\text{ if } \omega\in \mathcal{C} \cup \mathcal{E} \\
 0, &\text{ otherwise}
\end{cases}
\end{align}
Recall that the size of $\mathcal{I}$ is $m$, i.e., $m=|\mathcal{I}|$. Then, we have the following: 
\begin{align}
\label{probability_equation_1_bagging}
\text{Pr}(X \in \mathcal{E})= (\frac{m}{n})^k, \text{Pr}(X \in \mathcal{B})=1- (\frac{m}{n})^k, \text{ and } \text{Pr}(X \in \mathcal{C})=0. \\
\label{probability_equation_2_bagging}
\text{Pr}(Y \in \mathcal{E})= (\frac{m}{n^{\prime}})^k , \text{Pr}(Y \in \mathcal{C})=1- (\frac{m}{n^{\prime}})^k, \text{ and } \text{Pr}(Y \in \mathcal{B})=0.
\end{align}
We have $\text{Pr}(X \in \mathcal{E})=(\frac{m}{n})^k$  because each of the $k$ examples is sampled independently from $\mathcal{I}$ with probability $\frac{m}{n}$. Furthermore, since $\text{Pr}(X \in \mathcal{B})+ \text{Pr}(X \in \mathcal{E}) = 1$, we obtain $\text{Pr}(X \in \mathcal{B})=1- (\frac{m}{n})^k$. Since $X \not\sqsubseteq\mathcal{D}^{\prime}$, we have $\text{Pr}(X \in \mathcal{C})=0$. Similarly, we can compute the probabilities in~(\ref{probability_equation_2_bagging}). 

We assume $\underline{p_l}-\delta_l - (1- (\frac{m}{n})^k)\geq 0$.  We can make this assumption because we only need to find a sufficient condition for  $h(\mathcal{D}', \mathbf{x})=l$.  We define $\mathcal{B}^{\prime}\subseteq \mathcal{E}$, i.e., $\mathcal{B}^{\prime}$ is a subset of $\mathcal{E}$, such that we have the following: 
\begin{align}
\label{definition_of_aprime_bagging}
    \text{Pr}(X\in\mathcal{B}^{\prime}) = \underline{p_l} -\delta_l - \text{Pr}(X\in\mathcal{B}) = \underline{p_l} -\delta_l - (1- (\frac{m}{n})^k). 
\end{align}
We can find such subset because $\underline{p_l}-\delta_l$ is an integer multiple of $\frac{1}{n^{k}}$. Moreover, we define $\mathcal{R}$ as follows:
\begin{align}
\label{definition_of_r_bagging}
    \mathcal{R} = \mathcal{B}\cup \mathcal{B}^{\prime}.
\end{align}
Then, based on~(\ref{equation_of_condition_probability_bagging}), we have:
\begin{align}
  \text{Pr}(\mathcal{A}(X,\mathbf{x})=l)   \geq \underline{p_l}-\delta_l = \text{Pr}(X \in \mathcal{R}).
\end{align}
Therefore, we have the following: 
\begin{align}
\label{lemma_np_condition_ab_bagging}
  \text{Pr}(M(X)=1)=   \text{Pr}(\mathcal{A}(X,\mathbf{x})=l) \geq \text{Pr}(X \in \mathcal{R}). 
\end{align}
Furthermore, we have $\text{Pr}(X = \omega) > \gamma  \cdot \text{Pr}(Y = \omega)$ if and only if $\omega\in \mathcal{B}$ and $\text{Pr}(X = \omega) = \gamma  \cdot \text{Pr}(Y = \omega)$ if $\omega\in \mathcal{B}^{\prime}$, where $\gamma = (\frac{n^{\prime}}{n})^k$. Therefore, based on the definition of $\mathcal{R}$ in~(\ref{definition_of_r_bagging}) and the condition~(\ref{lemma_np_condition_ab_bagging}), we can apply Lemma~\ref{lemma_np} to obtain the following: 
\begin{align}
\text{Pr}(M(Y)=1)=   \text{Pr}(\mathcal{A}(Y,\mathbf{x})=l) \geq \text{Pr}(Y \in \mathcal{R}).
\end{align}
 $\text{Pr}(Y \in \mathcal{R})$ is a lower bound of $\text{Pr}(\mathcal{A}(Y,\mathbf{x})=l)$ and can be computed as follows: 
\begin{align}
 &\text{Pr}(Y \in \mathcal{R}) \\
 \label{probability_vk_in_r_1_bagging}
 =& \text{Pr}(Y \in \mathcal{B}) + \text{Pr}(Y \in \mathcal{B}^{\prime}) \\
  \label{probability_vk_in_r_2_bagging}
 =& \text{Pr}(Y \in \mathcal{B}^{\prime}) \\
  \label{probability_vk_in_r_3_bagging}
 =& \text{Pr}(X \in \mathcal{B}^{\prime})/\gamma  \\
  \label{probability_vk_in_r_4_bagging}
 =& \frac{1}{\gamma }\cdot (\underline{p_l}-\delta_l -  (1- (\frac{m}{n})^k)),
\end{align}
where we have~(\ref{probability_vk_in_r_2_bagging}) from~(\ref{probability_vk_in_r_1_bagging}) because $\text{Pr}(Y \in \mathcal{B})=0$,~(\ref{probability_vk_in_r_3_bagging}) from~(\ref{probability_vk_in_r_2_bagging}) because $\text{Pr}(X = \omega) = \gamma  \cdot \text{Pr}(Y = \omega)$ for $\omega\in \mathcal{B}^{\prime}$, and the last equation from~(\ref{definition_of_aprime_bagging}).

\myparatight{Deriving an upper bound of $\max_{j\neq l} \text{Pr}(\mathcal{A}(Y, \mathbf{x})=j)$} We define the following residual:
\begin{align}
\delta_j = (\lceil \overline{p}_{j} \cdot n^{k} \rceil)/n^{k} -\overline{p}_{j}, \forall j \in \{1,2,\cdots,c\}\setminus \{l\}. 
\end{align}
We leverage the second part of Lemma~\ref{lemma_np} to derive such an upper bound. We assume $\text{Pr}(X \in\mathcal{E}) \geq \overline{p}_{j}+\delta_j$, $\forall j \in \{1,2,\cdots,c\}\setminus \{l\}$. We can make the assumption because we derive a sufficient condition for $h(\mathcal{D}', \mathbf{x})=l$.
For $\forall j \in \{1,2,\cdots,c\}\setminus \{l\}$, we define $\mathcal{C}_j\subseteq \mathcal{E}$ such that we have the following: 
\begin{align}
\label{c_j_condition_in_certify}
    \text{Pr}(X \in\mathcal{C}_j) = \overline{p}_{j} + \delta_j. 
\end{align}
We can find such $\mathcal{C}_j$ because $\overline{p}_{j} + \delta_j$ is an integer multiple of $\frac{1}{n^{k}}$.
Moreover, we define the following space: 
\begin{align}
\label{definition_of_q_bagging}
    \mathcal{Q}_j = \mathcal{C}\cup \mathcal{C}_j.
\end{align}
Therefore, based on~(\ref{equation_of_condition_probability_bagging}), we have:
\begin{align}
  \text{Pr}(\mathcal{A}(X,\mathbf{x})=j)   \leq \overline{p}_j + \delta_j = \text{Pr}(X \in \mathcal{Q}_j). 
\end{align}
We define a function $M_j(\omega)= \mathbb{I}(\mathcal{A}(\omega,\mathbf{x})=j)$, where $\omega\in \Omega$. Based on Lemma~\ref{lemma_np}, we have the following: 
\begin{align}
    \text{Pr}(M(Y)=1)=   \text{Pr}(\mathcal{A}(Y,\mathbf{x})=j) \leq \text{Pr}(Y \in \mathcal{Q}_j),
\end{align}
where $\text{Pr}(Y \in \mathcal{Q}_j)$ can be computed as follows: 
\begin{align}
 &\text{Pr}(Y \in \mathcal{Q}_j) \\
 =& \text{Pr}(Y \in \mathcal{C}) + \text{Pr}(Y \in \mathcal{C}_j) \\
 =& 1- (\frac{m}{n^{\prime}})^k + \text{Pr}(Y \in \mathcal{C}_j) \\
 =& 1- (\frac{m}{n^{\prime}})^k + \text{Pr}(X \in \mathcal{C}_j)/\gamma  \\
 =& 1- (\frac{m}{n^{\prime}})^k + \frac{1}{\gamma }\cdot (\overline{p}_j + \delta_j).
\end{align}
Therefore, we have: 
\begin{align}
 &\max_{j\neq l} \text{Pr}(\mathcal{A}(Y, \mathbf{x})=j)\\
 \leq & \max_{j\neq l} \text{Pr}(Y \in \mathcal{Q}_j) \\
 =& 1- (\frac{m}{n^{\prime}})^k + \frac{1}{\gamma }\cdot \max_{j\neq l} (\overline{p}_j +\delta_j) \\
 \leq & 1- (\frac{m}{n^{\prime}})^k + \frac{1}{\gamma }\cdot (\overline{p}_s + \delta_s),
\end{align}
where $\overline{p}_s + \delta_s \geq \max\limits_{j \neq l} (\overline{p}_j + \delta_j)$. 

\myparatight{Deriving the certified poisoning size} To reach the goal $\text{Pr}(\mathcal{A}(Y,\mathbf{x})=l) > \max\limits_{j\neq l}\text{Pr}(\mathcal{A}(Y,\mathbf{x})=j)$, it is sufficient to have the following:   
\begin{align}
  &  \frac{1}{\gamma }\cdot (\underline{p_l}-\delta_l - (1- (\frac{m}{n})^k)) 
> 1- (\frac{m}{n^{\prime}})^k + \frac{1}{\gamma }\cdot (\overline{p}_s +\delta_s) \\
 \Longleftrightarrow &  (\frac{n^{\prime}}{n})^k -2\cdot (\frac{m}{n})^k + 1 -(\underline{p_l} -\overline{p}_s  -\delta_l - \delta_s) < 0.
\end{align}
Taking all poisoned training datasets $\mathcal{D}'$  (i.e., $n-r \leq n^{\prime} \leq  n+r$) into consideration, we have the following sufficient condition: 
\begin{align}
\label{certify_poisoning_size_condition_proof_final_check}
  \max_{n-r \leq n^{\prime} \leq n+r} (\frac{n^{\prime}}{n})^k -2\cdot (\frac{m}{n})^k + 1 -(\underline{p_l} -\overline{p}_s - \delta_l - \delta_s) < 0. 
\end{align}
Note that $m = \max(n,n^{\prime}) - r$. Furthermore, when the above condition~(\ref{certify_poisoning_size_condition_proof_final_check}) is satisfied, we  have $\underline{p_l}-\delta_l - (1- (\frac{m}{n})^k)\geq 0$ and $\text{Pr}(X \in\mathcal{E}) =(\frac{m}{n})^k \geq \overline{p}_{j}+\delta_j,\forall j \in \{1,2,\cdots,c\}\setminus \{l\}$, which are the conditions when we can construct the spaces $\mathcal{B}^{\prime}$ and $\mathcal{C}_j$. The certified poisoning size $r^{*}$ is the maximum $r$ that satisfies the above sufficient condition. 
In other words, our certified poisoning size $r^{*}$ is the solution to the following optimization problem: 
\begin{align}
r^{*} & =  \argmax_{r} r \nonumber \\
\label{certified_condition_bagging_proof}
& \text{s.t.} \max_{n-r \leq n' \leq n+r}(\frac{n^{\prime}}{n})^{k} - 2 \cdot (\frac{\max(n,n^{\prime})-r}{n})^{k} + 1 - (\underline{p_l}-\overline{p}_s -\delta_l - \delta_s) < 0.  
\end{align}

\section{Proof of Theorem~\ref{tightness_theorem}}
\label{proof_of_tightness}
Our idea is to construct a learning algorithm $\mathcal{A}^{*}$ such that the label $l$ is not predicted by the bagging predictor or there exist ties. 
When $r>r^*$ and $\delta_l = \delta_s = 0$, there exists a poisoned training dataset $\mathcal{D}^{\prime}$ with a certain $n^{\prime} \in [n-r, n+r]$ such that we have:
\begin{align}
\label{equation_for_d_construct}
 &   (\frac{n^{\prime}}{n})^{k} - 2 \cdot (\frac{\max(n,n^{\prime})-r}{n})^{k} + 1 - (\underline{p_l}-\overline{p}_s ) \geq 0 \\
\Longleftrightarrow  &    (\frac{n^{\prime}}{n})^{k} - 2 \cdot (\frac{m}{n})^{k} + 1 - (\underline{p_l}-\overline{p}_s ) \geq 0 \\
\Longleftrightarrow &  1+(\frac{n^{\prime}}{n})^k -2\cdot (\frac{m}{n})^k \geq  \underline{p_l} -\overline{p}_s   \\
\label{equation_for_d_construct_equivalent}
\Longleftrightarrow & \frac{1}{\gamma }\cdot (\underline{p_l} - (1- (\frac{m}{n})^k) \leq 1-(\frac{m}{n^{\prime}})^k + \frac{1}{\gamma }\cdot \overline{p}_s, 
\end{align}
where $m=\max(n,n^{\prime})-r$ and $\gamma  = (\frac{n^{\prime}}{n})^k$. We let $\mathcal{Q}_s = \mathcal{C}\cup \mathcal{C}^{\prime}_s$, where $\mathcal{C}^{\prime}_s$ satisfies the following: 
\begin{align}
\mathcal{C}^{\prime}_s \subseteq \mathcal{E},\mathcal{C}^{\prime}_s \cap \mathcal{B}^{\prime} = \emptyset, \text{ and } \text{Pr}(X\in \mathcal{C}^{\prime}_s )=\overline{p}_s.
\end{align}
Note that we can construct such $\mathcal{C}^{\prime}_s$ because $\underline{p_l} +\overline{p}_s  \leq 1$. Then, we  divide the remaining space $\Omega\setminus (\mathcal{R}\cup \mathcal{Q}_s)$ into $c-2$ subspaces such that $\text{Pr}(X \in \mathcal{Q}_j) \leq \overline{p}_s $, where $j\in \{1,2,\cdots,c\}\setminus\{l,s\}$.  We can construct such subspaces because $\underline{p_l} +(c-1)\cdot \overline{p}_s  \geq 1$. Then, based on these subspaces, we  construct the following learning algorithm: 
\begin{align}
\mathcal{A}^{*}(\omega,\mathbf{x})  =
\begin{cases}
& l, \text{ if } \omega \in \mathcal{R} \\
& j, \text{ if } \omega \in \mathcal{Q}_j
\end{cases}
\end{align}
Then, we have the following based on the above definition of the learning algorithm $\mathcal{A}^{*}$: 
\begin{align}
    & \text{Pr}(\mathcal{A}^{*}(X,\mathbf{x})=l)=\text{Pr}(X\in \mathcal{R}) =\underline{p_l} \\
    & \text{Pr}(\mathcal{A}^{*}(X,\mathbf{x})=s)=\text{Pr}(X\in \mathcal{Q}_s) =\overline{p}_s  \\
    & \text{Pr}(\mathcal{A}^{*}(X,\mathbf{x})=j)=\text{Pr}(X\in \mathcal{Q}_j) \leq \overline{p}_s , j\in \{1,2,\cdots,c\}\setminus\{l,s\}.
\end{align}
Therefore, the learning algorithm $\mathcal{A}^{*}$ is consistent with~(\ref{equation_of_condition_probability_bagging}). Next, we show that $l$ is  not predicted by the bagging predictor or there exist ties when the training dataset is $\mathcal{D}^{\prime}$.  In particular, we have the following: 
\begin{align}
   & \text{Pr}(\mathcal{A}^{*}(Y,\mathbf{x})=l) \\
     = & \text{Pr}(Y \in \mathcal{R}) \\
     \label{construct_final_compare_1}
     =& \frac{1}{\gamma }\cdot (\underline{p_l} - (1-(\frac{m}{n})^k)) \\
     \label{construct_final_compare_2}
     \leq & 1- (\frac{m}{n^{\prime}})^k + \frac{1}{\gamma }\cdot \overline{p}_s  \\
     = & \text{Pr}(Y \in \mathcal{Q}_s) \\
     =& \text{Pr}(\mathcal{A}^{*}(Y,\mathbf{x})=s),
\end{align}
where $\gamma  = (\frac{n^{\prime}}{n})^k$ and we have~(\ref{construct_final_compare_2}) from~(\ref{construct_final_compare_1}) because of~(\ref{equation_for_d_construct_equivalent}). Therefore, label $l$ is not predicted for $\mathbf{x}$ or there exist ties when the training dataset is $\mathcal{D}^{\prime}$. 

\section{Proof of Theorem~\ref{proposition_of_certify}}
\label{proof_of_certify_proposition}
 Based on the definition of \textsc{SimuEM} in \citep{jia2020certified}, we have:  
\begin{align}
    \text{Pr}(\underline{p_{l_i}} \leq \text{Pr}(\mathcal{A}(g(\mathcal{D}),\mathbf{x}_i)=l_i) \land \overline{p}_{j}\geq \text{Pr}(\mathcal{A}(g(\mathcal{D}),\mathbf{x}_i)=j), \forall j \neq l_i) \geq 1 - \frac{\alpha}{e}.
\end{align}
Therefore, the probability that \textsc{Certify} returns an incorrect certified poisoning size for a testing example $\mathbf{x}_i$ is at most  $\frac{\alpha}{e}$, i.e., we have: 
\begin{align}
    \text{Pr}((\exists \mathcal{D}'\in B(\mathcal{D}, r_i^*), h(\mathcal{D}',\mathbf{x}_i)\neq l_i)|l_i \neq \text{ABSTAIN} ) \leq \frac{\alpha}{e}. 
\end{align}
Then, we have the following: 
\begin{align}
& \text{Pr}(\cap_{\mathbf{x}_i \in \mathcal{D}_e} ((\forall \mathcal{D}'\in B(\mathcal{D}, r_i^*), h(\mathcal{D}',\mathbf{x}_i)=l_i)|l_i \neq \text{ABSTAIN} )) \\
\label{theorem_4_apply_booles_inequality_1}
&= 1 - \text{Pr}(\cup_{\mathbf{x}_i \in \mathcal{D}_e} ((\exists \mathcal{D}'\in B(\mathcal{D},r_i^*), h(\mathcal{D}',\mathbf{x}_i) \neq l_i)|l_i \neq \text{ABSTAIN} )) \\
\label{theorem_4_apply_booles_inequality_2}
& \geq  1 - \sum_{\mathbf{x}_i \in \mathcal{D}_e}\text{Pr}((\exists \mathcal{D}'\in B(\mathcal{D}, r_i^*), h(\mathcal{D}',\mathbf{x}_i) \neq l_i)|l_i \neq \text{ABSTAIN} ) \\
& \geq  1- e \cdot \frac{\alpha}{e} \\
& = 1 -\alpha 
\end{align}
We have (\ref{theorem_4_apply_booles_inequality_2}) from (\ref{theorem_4_apply_booles_inequality_1}) according to the Boole's inequality.

\section{Derivation of the Analytical Form of $n'$}

We define $L(n^{\prime})$ as follows:
\begin{align}
  L(n^{\prime})=  (\frac{n^{\prime}}{n})^{k} - 2 \cdot (\frac{\max(n,n^{\prime})-r}{n})^{k} + 1 - (\underline{p_l}-\overline{p}_s -\delta_l -\delta_s).
\end{align}
We have the following analytical form of $n'$ at which $L(n^{\prime})$ reaches its maximum: 
{\small
\begin{align}
   \argmax_{n-r \leq n' \leq n+r} L(n') 
  \label{analyticaln}
   =\begin{cases}
     n, & \text{ if }  r\leq n\cdot (1 - \sqrt[k-1]{\frac{1}{2}})  \\ 
     n+r,  & \text{ if } r\geq n\cdot (\sqrt[k-1]{2}-1) \\
     \lceil \frac{r}{1 - \sqrt[k-1]{\frac{1}{2}}} \rceil \text{ or } \lfloor \frac{r}{1 - \sqrt[k-1]{\frac{1}{2}}} \rfloor,  & \text{ otherwise. }  
    \end{cases}
\end{align}
}
Next, we show details on how to derive such analytical form of  $n'$.
When $n-r \leq n' \leq n$, we have the following: 
\begin{align}
    L(n^{\prime})= (\frac{n^{\prime}}{n})^{k} - 2 \cdot (\frac{n-r}{n})^{k} + 1 - (\underline{p_l}-\overline{p}_s -\delta_l -\delta_s).
\end{align}
Therefore, when $n-r \leq n' \leq n$,  $L(n^{\prime})$ increases as $n^{\prime}$ increases. Thus, $L(n^{\prime})$ reaches its maximum value when $n \leq n' \leq n+r$. When $n \leq n' \leq n+r$, we have the following:
\begin{align}
    L(n^{\prime})= (\frac{n^{\prime}}{n})^{k} - 2 \cdot (\frac{n^{\prime}-r}{n})^{k} + 1 - (\underline{p_l}-\overline{p}_s -\delta_l -\delta_s).
\end{align}
Moreover, we have:
\begin{align}
    & \frac{\partial L(x)}{\partial x} \\
    =& \frac{1}{n^{k}}\cdot (k\cdot x^{k-1}-2\cdot k \cdot (x-r)^{k-1}) \\
    =& \frac{k\cdot x^{k-1}}{n^{k}}\cdot ( 1 - 2 \cdot (1 - \frac{r}{x})^{k-1}).
\end{align}
 $\frac{k\cdot x^{k-1}}{n^{k}}$ is larger than $0$. Moreover, $1 - 2 \cdot (1 - \frac{r}{x})^{k-1}$ decreases as $x$ increases when $x\geq r$ and it only has one root that is no smaller than $r$ which is as follows:
\begin{align}
    x_{root} = \frac{r}{1 - \sqrt[k-1]{\frac{1}{2}}}.
\end{align}
Therefore, we have $\frac{\partial L(x)}{\partial x} > 0$ when $r\leq x < \frac{r}{1 - \sqrt[k-1]{\frac{1}{2}}}$ and $\frac{\partial L(x)}{\partial x} < 0$ when $x > \frac{r}{1 - \sqrt[k-1]{\frac{1}{2}}}$.  $L(x)$ increases as $x$ increases in the range $[r,\frac{r}{1 - \sqrt[k-1]{\frac{1}{2}}})$ and decreases as $x$ increases in the range $(\frac{r}{1 - \sqrt[k-1]{\frac{1}{2}}},+\infty)$. Therefore, we have the following three cases:

\myparatight{Case I}  When $r \leq n\cdot (1 - \sqrt[k-1]{\frac{1}{2}})$,  $L(n^{\prime})$ reaches its maximum value at $n^{\prime}=n$ since $L(n^{\prime})$ decreases as $n^{\prime}$ increases in the range $[n,n+r]$. 

\myparatight{Case II} When $n\cdot (1 - \sqrt[k-1]{\frac{1}{2}}) < r <n\cdot (\sqrt[k-1]{2}-1)$,  $L(n^{\prime})$ reaches its maximum value at $n^{\prime}=\lceil \frac{r}{1 - \sqrt[k-1]{\frac{1}{2}}} \rceil \text{ or } \lfloor \frac{r}{1 - \sqrt[k-1]{\frac{1}{2}}} \rfloor$ since $L(n^{\prime})$ increases as $n^{\prime}$ increases in the range $[n,  \frac{r}{1 - \sqrt[k-1]{\frac{1}{2}}}]$ and decreases as $n^{\prime}$ increases in the range $[\frac{r}{1 - \sqrt[k-1]{\frac{1}{2}}}, n+r]$. 

\myparatight{Case III} When $r \geq n\cdot (\sqrt[k-1]{2}-1)$,  $L(n^{\prime})$ reaches its maximum value at $n^{\prime}=n+r$ since $L(n^{\prime})$ increases as $n^{\prime}$ increases in the range $[n,n+r]$.